\documentclass{statsoc}

\usepackage[a4paper]{geometry}
\usepackage{amssymb}
\usepackage{enumerate}
\usepackage{pgf,tikz,tikz-cd}
\usepackage{graphicx}
\usepackage{multirow}
\usepackage{listings}
\usepackage{epstopdf}
\usepackage{mathrsfs}
\usepackage{natbib}
\usepackage[colorlinks=true,linkcolor=blue, allcolors=blue]{hyperref}
\usepackage{url}
\usepackage{algorithm}
\usepackage{ulem}
\usepackage[noend]{algpseudocode}
\usepackage{MnSymbol}
\usepackage{comment}
\usepackage{array}
\usepackage{float}

\newcommand{\CC}{\mathbb{C}}
\newcommand{\RR}{\mathbb{R}}
\newcommand{\NN}{\mathbb{N}}

\newcommand{\ZZ}{\mathbb{Z}}
\newcommand{\EE}{\mathbb{E}}

\newcommand{\PP}{\mathbb{P}}

\newcommand{\dd}{\,\mathrm{d}}

\newcommand{\iid}{\overset{iid}{\sim}}

\newtheorem{thm}{Theorem}

\newtheorem{prop}[thm]{Proposition}
\newtheorem{eg}[thm]{Example}

\newtheorem{Lem}[thm]{Lemma}

\newcommand{\Gmax}{G_{\max}(f, \mathcal{G} )}

\usepackage{xcolor}

\numberwithin{equation}{section}
\numberwithin{thm}{section}

\title[Estimating Maximal Symmetries]{Estimating Maximal Symmetries of Regression Functions via Subgroup Lattices}

\author{Louis G. Christie}
\address{University of Cambridge,
			Cambridge,
			United Kingdom.}
\email{lgc26@cam.ac.uk}

\author[L. G. Christie \& J. A. D. Aston]{John A. D. Aston}
\address{University of Cambridge,
			Cambridge,
			United Kingdom.}
\email{j.aston@statslab.cam.ac.uk}

\begin{document}

\begin{abstract}
We present a method for estimating the maximal symmetry of a continuous regression function. Knowledge of such a symmetry can be used to significantly improve modelling by removing the modes of variation resulting from the symmetries. Symmetry estimation is carried out using hypothesis testing for invariance strategically over the subgroup lattice of a search group $\mathcal{G}$ acting on the feature space. We show that the estimation of the unique maximal invariant subgroup of $\mathcal{G}$ generalises useful tools from linear dimension reduction to a non linear context. We show that the estimation is consistent when the subgroup lattice chosen is finite, even when some of the subgroups themselves are infinite.  We demonstrate the performance of this estimator in synthetic settings and apply the methods to two data sets: satellite measurements of the earth's magnetic field intensity; and the distribution of sunspots.
\end{abstract}


\section{Introduction}

Many objects we wish to model in statistics obey symmetries. Time series can be seasonal \citep{box2015time}, which is a form of discrete translation symmetry. In biology, viral capsids can exhibit icosahedral symmetries, as shown in figure \ref{fig:protiens} below \citep{jiang2017atomic}. Linear models can be invariant to one or more of specific features, which is equivalent to a continuous translation symmetry along that feature's axis as in figure \ref{fig:demo_alg}. In statistical shape analysis \citep{kendall1984shape}, we consider collections of $k$ landmark points in $\RR^m$ under the symmetries of rotation, scaling, and translation simultaneously. Graph neural networks utilise the permutation symmetries of nodes and edges \citep{atz2021geometric}. \\
%
\begin{figure}[h!]
	\centering
	\begin{tabular}{ >{\centering}p{4cm}  >{\centering}p{4cm}  >{\centering}p{4cm}}
				\includegraphics[scale = 0.15]{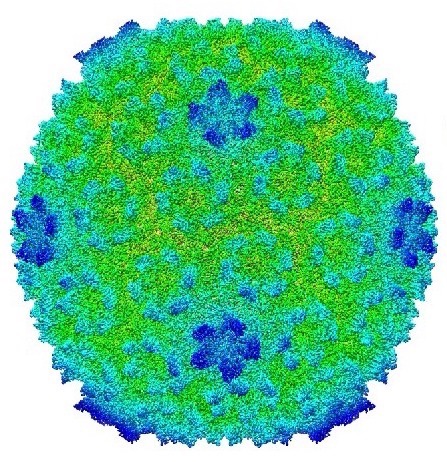} &
				\includegraphics[scale = 0.15]{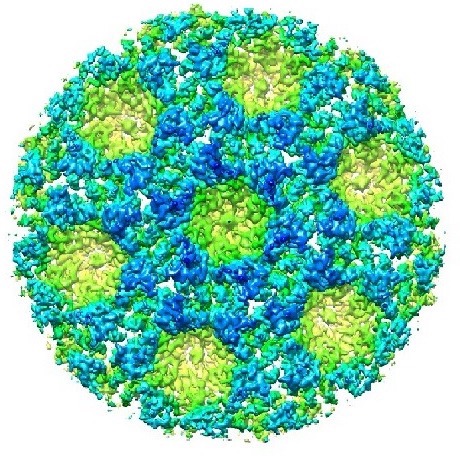} &
				\includegraphics[scale=0.085]{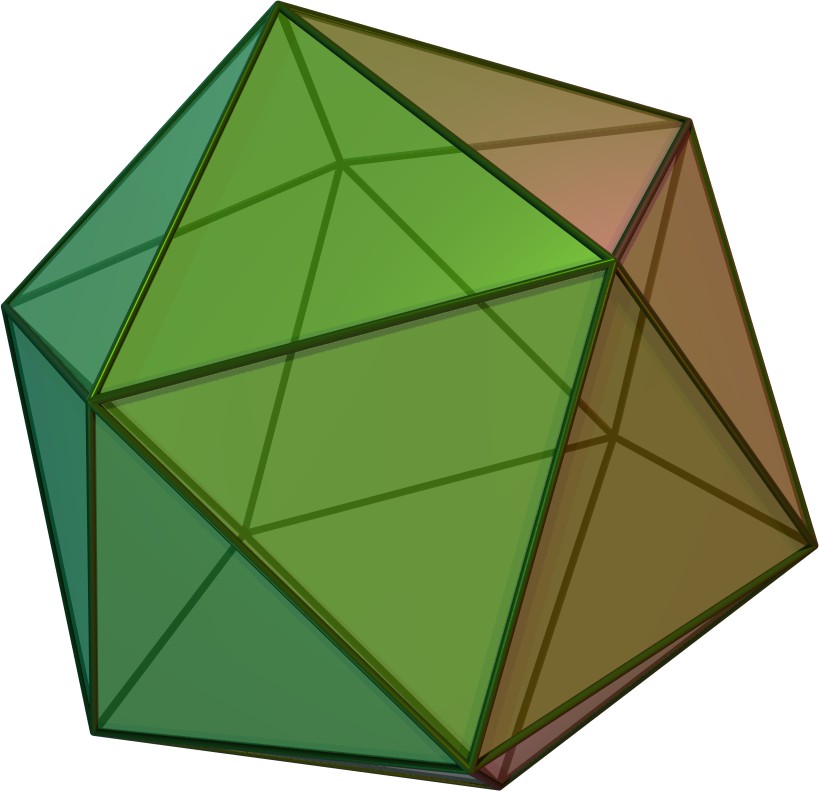} 
				\tabularnewline			
				(a) Sf6 Capsid & (b) HIV-1 CA Capsid & (c)  Icosahedron
	\end{tabular}

	\caption{Subfigures (a) and (b) show Cryo-EM reconstructions of representative viruses (Figure 1 in \cite{jiang2017atomic}). Figure (c) shows an icosahedron - a platonic solid with the same symmetries as the viruses.  }
	\label{fig:protiens}
\end{figure}

In general, consider estimating a continuous square integrable function $f : \mathcal{X} \rightarrow \RR$ with an invariance property:
\begin{equation}
	\label{eq:inv_prop}
 f( \phi( x ) ) = f(x)  
\end{equation}
for all $x \in \mathcal{X}$ and all invertible transformations $\phi : \mathcal{X} \rightarrow \mathcal{X}$ in some set $\Phi$. We say that such an $f$ is \textbf{strongly $\Phi$-invariant}. Any such set $\Phi$ must naturally have the properties that:
\begin{equation}
	\text{(G1) } \phi \circ \psi \in \Phi \text{ for all } \psi, \phi \in \Phi; \text{ and } \quad \text{ (G2) } \phi^{-1} \in \Phi \text{ for all } \phi \in \Phi,
\end{equation}
because we can use \ref{eq:inv_prop} to check that for all $\phi, \psi \in \Phi$ we have:
\begin{align}
	f  (\phi (\psi(x) ) ) &= f ( \psi(x) ) = f (x)  \\
	f ( \phi^{-1} (x) ) &= f ( \phi ( \phi^{-1} (x) ) ) = f(x) 
\end{align}
This means that the set of transformations $\Phi$ is an example of a \textit{group}, defined in section \ref{sec:background}, though we will avoid technical group theory in this introduction. This will be important later as the properties of groups are used both for statistical and computational benefit.  \\
\\
Suppose that we have collected data $\mathcal{D} = \{ (X_i, Y_i) \}_{i = 1}^n $ with each $(X_i, Y_i) \in \mathcal{X} \times \RR$, $\EE( Y_i \mid X_i ) = f(X_i)$, and $\mathrm{Var}( Y_i \mid X_i) = \sigma^2 < \infty$. The basis of this paper will be such regression contexts with independent identically distributed data. This covers classification situations, either using an indicator function or a probabilistic regression function. It can also describe other functional objects, for example the map $V : \RR^3 \rightarrow \RR_{\geq 0}$ that describes the mass density of a viral capsid $V$, or the relationship of conditional means of one coordinate of a joint vector.   \\
\\
There are three main methods for including the information of the symmetries of $\Phi$ in an estimator $\hat{f}$ of $f$:
\begin{itemize}
	\item Firstly, we can sample and apply transformations to construct a new data set $\mathcal{D}' = \{ ( \phi_{ij} ( X_i ) , Y_i ) : j \in \{0, \dots, J_i \} \}_{i=1}^n $ for some randomly or deterministically chosen $\phi_{ij} \in \Phi$ (known in machine learning as \textit{data augmentation} \citep{shorten2019survey, chen2020group}, which is distinct from Bayesian data augmentation \citep{tanner1987calculation, van2001art}). 
	\item Secondly, we can consider projecting data points (\textit{data projection}) in the quotient space $\mathcal{X} / \Phi$, which is the space of equivalence classes under the relation $X_i \sim X_j$ if and only if $X_i = \phi( X_j )$ for some $\phi \in \Phi$. This quotient space is often represented by a set of unique points from each class, obtained by moving the data points using the transformations from $\Phi$.
	\item  Lastly, we can transform the estimate $\hat{f}$ into a $\Phi$-invariant estimate by taking $S_\Phi \hat{f} = \EE ( \hat{f}( \phi (x), \mathcal{D} ) \mid \mathcal{D})$ (where $\phi$ is a uniform random variable on $\Phi$, which exists under certain topological conditions on $\Phi$, see section \ref{ssec:group_theory}), known as \textit{feature averaging} \citep{elesedy2021provablyLin}.
\end{itemize}

These methods for particular transformations have each been used in numerous statistical contexts. Shape statistics uses data projection where we have explicit projection maps, while the Cryo-EM images are improved by feature averaging. These methods and other symmetrisation techniques have been shown to be broadly applicable in many statistical and machine learning contexts; a good overview of these contexts is \citet{bronstein2021geometric}. \\
\\
Recent studies have quantified the benefit of these methods, such as in \citet{lyle2020benefits, bietti2021sample, elesedy2021provably, huang2022quantifying}. In summary, one can prove that such symmetrisation techniques outperform the base estimator, almost surely over the data. For example, in the case of feature averaging under topological conditions on the transformation set $\Phi$ we have:
\begin{equation}
	\label{eq:consistency}
	\| f - S_\Phi \hat{f} \|_2 = \| S_\Phi f - S_\Phi \hat{f} \|_2 \leq \| S_\Phi \| \| f - \hat{f} \|_2 = \| f - \hat{f} \|_2
\end{equation}
using the $\Phi$-invariance of $f$, which implies $f = S_\Phi f$, and the fact that $S_\Phi$ is a projection (i.e., an idempotent linear operator) \citep{elesedy2021provablyLin} and thus has operator norm at most $1$. This can be seen as averaging out some modes of variation in the estimate $\hat{f}$ in a flexible non-linear manner, or by reducing the entropy of the function class we are estimating over. A simple example of this is averaging a function $f$ on $\RR^3$ with the symmetry of the special orthogonal group $SO(3) = \{  x \mapsto Ax : A \in \RR^{3 \times 3}, A^T A = I, \det(A) = 1 \}$ which describes rotations around the origin in 3D space. In this case the function $S_{SO(3)} f$ varies only along the radial direction at any point $x \in \RR^3$ and is constant in the orthogonal directions. This situation occurs for spherically symmetric density functions as in \citet{baringhaus1991testing, garcia2020optimal}. \\
\\
Data projection can be thought of as picking a suitable non-linear transformation of the covariates. In the $SO(3)$ example above, we can project the data points to a ray emanating from the origin: $x \mapsto ( \| x \|_2 , 0 ,0 )$. Estimation of $f$ then becomes much simpler knowing that the inputs all lie along a one dimensional subspace of $\RR^3$. \\
\\
This of course relies heavily on the assumption that $f$ truly is $\Phi$-invariant. Otherwise we create obvious problems: we will struggle to identify images of the digits 6 or 9 if we assume that their classification functions are invariant to any rotations, for example. More generally, if $\hat{f}$ is a consistent estimator of $f$ then $S_\Phi \hat{f}$ converges in probability to $S_\Phi f$ (from equation \ref{eq:consistency}). This means that inappropriately used symmetries (i.e., when $f$ is not $\Phi$-invariant) can create asymptotic bias (of $S_\Phi f - f \neq 0$). Therefore, we want to use as much symmetry as possible without using too much. \\
\\
There have been a few efforts to develop methods that use the data to select appropriate symmetries, such as \citet{cubuk2019autoaugment, lim2019fast, benton2020learning}. These methods lose the global structure by estimating sub\textit{sets} of transformations rather than sub\textit{groups}. Moreover, little if any attention has been given to the statistical properties of these methods. Other methods in statistics have focused on testing for the symmetries of  distributions \citep{baringhaus1991testing, garcia2020optimal}, but these methods cannot be used for this problem because the conditional means can exhibit symmetry even if the marginal distribution of the covariates does not. \\
\\
Thus in this paper we present a statistical method for establishing symmetries of an object, to be used when estimating it. This first requires us to develop an inferential framework for group objects. Unlike other work that does this for group \textit{valued} random objects \citep{grenander2008probabilities}, in our framework the objects are the entire groups themselves. With this framework it is possible to solve a particular problem - estimating the maximal symmetries of regression functions.   \\
\\
This framework is an explicit generalisation of some techniques in variable subset selection \citep{mallows2000some, berk1978comparing}. Consider a linear regression function $f : \RR^p \rightarrow \RR$ given by $f(X) = X^T \beta$ for which some features $i \in I$ have $\beta_i = 0$, where $\beta_i$ are the regression coefficients in $\beta$. We could estimate $I$ by considering subsets $A \subseteq \{ 1, \dots, p \}$ and testing whether $\beta_A = 0$ (where the set subscript indicates the vector $( \beta_a : a \in A)$). These subsets can be structured in a \textit{lattice} (see section \ref{ssec:lattices} and figure \ref{fig:demo_alg}), a directed graph with a node for each subset and an edge from $A \rightarrow B$ if $B \subseteq A$ (though we omit the directions from the figures as they always point down in this paper). This is important because it shows how we can avoid running tests for some of the subsets $A$. If we test at the bottom of the lattice first, any rejection immediately tells us about the larger subsets containing the smaller subset so we can avoid testing for invariance to these subsets unnecessarily. \\
\\
This linear example is in fact an explicit special case of the symmetries considered here. In this case the domain is $\mathcal{X} = \RR^p$, and the symmetries are translations $\Phi_{\mathbb{T}} = \{ x \mapsto x + u : u \in \RR^p \}$. If $f$ is invariant to translations along a particular axis (in the standard orthonormal basis $\{ e_i \}$) then $\beta ( X + a e_i ) = \beta X$ for all $a \in \RR$, which is equivalent to $\beta_i = 0$. This shows that variable selection in linear regression is precisely a symmetry estimation problem. There are of course many other sub-symmetries of this translation action, for example translations along the line from the origin through $(1,\dots, 1)$. This is the idea of sufficient dimension reduction \citep{li1991sliced}, which is also a special case of this symmetry estimation problem.  \\
\\
So in general, the question we seek to answer is: ``Given a large set of transformations $\Phi$ (that satisfies (G1) and (G2)), what is the largest subset of transformations $\Psi$ that satisfy equation \ref{eq:inv_prop} for our particular regression $f$?'' In the special cases above, $\Phi$ is either finite (as in variable selection) or linear (as in sufficient dimension reduction), but in our case it can be infinite and non-linear. To answer this, we will be using the idea that $\Psi$ must also satisfy the group properties (G1) and (G2) (and so must be a \textit{subgroup} of $\Phi$) and use the lattice structure of these subgroups to allow for the lattice testing regime as above.

\begin{figure}[h]
	\centering
	\begin{tabular}{ccc}
	
	\begin{tikzpicture}[scale = 0.8]
	\draw[fill] (0,0) circle (0.07cm) node[right]{$\varnothing$};		
	\draw[fill] (-2,2) circle (0.07cm) node[right]{$\{ 1 \}$};
	\draw[fill] (0,2) circle (0.07cm) node[right]{$\{ 2\}$};
	\draw[fill] (2,2) circle (0.07cm) node[right]{$\{ 3 \}$};
	\draw[fill] (-2,4) circle (0.07cm) node[right]{$\{ 1, 2 \}$};
	\draw[fill] (0,4) circle (0.07cm) node[right]{$\{ 1,3 \}$};
	\draw[fill] (2,4) circle (0.07cm) node[right]{$\{ 2,3  \}$};		
	\draw[fill] (0,6) circle (0.07cm) node[right]{$\{ 1,2 ,3 \}$};
	
	\draw (0,0) -- (-2,2);
	\draw (0,0) -- (0,2);
	\draw (0,0) -- (2,2);
	
	\draw (-2,2) -- (-2, 4);
	\draw (-2,2) -- (0, 4);
	\draw (0,2) -- (-2, 4);
	\draw (0,2) -- (2, 4);
	\draw (2,2) -- (0, 4);		
	\draw (2,2) -- (2, 4);
	
	\draw (-2,4) -- (0, 6);
	\draw (0,4) -- (0, 6);
	\draw (2,4) -- (0, 6);
	\end{tikzpicture}
	
	&
	
	$\phantom{00}$
	
	&
	\begin{tikzpicture}[scale = 0.8]
	\draw[fill] (0,0) circle (0.07cm) node[right]{$I$};
	\draw[fill] (-2,2) circle (0.07cm) node[right]{$\RR_1$};
	\draw[fill] (0,2) circle (0.07cm) node[right]{$\RR_2$};
	\draw[fill] (2,2) circle (0.07cm) node[right]{$\RR_3$};
	\draw[fill] (-2,4) circle (0.07cm) node[right]{$\langle \RR_1, \RR_2 \rangle $};
	\draw[fill] (0,4) circle (0.07cm) node[right]{$\langle \RR_1, \RR_3 \rangle$};
	\draw[fill] (2,4) circle (0.07cm) node[right]{$\langle \RR_2, \RR_3 \rangle$};		
	\draw[fill] (0,6) circle (0.07cm) node[right]{$\RR^3$};
	
	\draw (0,0) -- (-2,2);
	\draw (0,0) -- (0,2);
	\draw (0,0) -- (2,2);
	
	\draw (-2,2) -- (-2, 4);
	\draw (-2,2) -- (0, 4);
	\draw (0,2) -- (-2, 4);
	\draw (0,2) -- (2, 4);
	\draw (2,2) -- (0, 4);		
	\draw (2,2) -- (2, 4);
	
	\draw (-2,4) -- (0, 6);
	\draw (0,4) -- (0, 6);
	\draw (2,4) -- (0, 6);
	
	\draw[fill, gray] (0.8,1.6) circle (0.07cm) node[below right]{$\langle (0,1,1) \rangle$};
	\draw[fill, gray] (1.5,3) circle (0.07cm) node[below right]{$\langle (0,1/2,1/2) \rangle$};
	\draw[dashed, gray] (0,0) --(2,4);
	\end{tikzpicture} \\
	
	(a) Lattice of subsets $A \subseteq \{ 1, 2, 3 \}$. & & (b) Lattice of (some) subgroups of $\RR^p$.
		
	\end{tabular}

	\caption{Variable subset selection lattice of subsets, and symmetry estimation lattice of subgroups. The symmetry here is the translation action of $\Phi_{\mathbb{T}}$ on $\RR^3$, and we have taken a finite sub-lattice of the lattice of closed subgroups where $\RR_i$ is the translation group along the $i^{th}$ orthonormal basis vector of $\RR^3$. We have also included gray and dashed edges to other symmetries such as $\{ x \mapsto x + a( e_2 + e_3) : a \in \ZZ \}$ and the larger symmetry $\{ x \mapsto x + (a/2)( e_2 + e_3 ) : a \in \ZZ \}$. }
	\label{fig:demo_alg}
\end{figure}
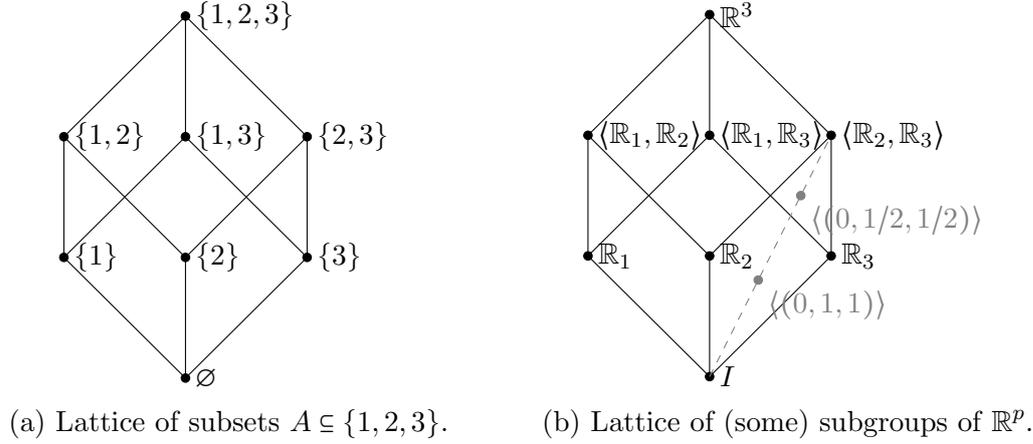

\subsection{Our Contributions}

We first present an inferential framework for the space of subgroups of some topological group $\mathcal{G}$. This is based on the ordering of the subgroups into a \textit{lattice}, an algebraic object that describes the relationships between the subgroups and is defined in section \ref{sec:inf_framework}. We then apply this framework to give a method for estimating the best symmetry to use from $\mathcal{G}$ when estimating a regression function $f$. We also show that this estimate falls within a $1 - \alpha$ confidence region in the chosen lattice for the true symmetry. This is done by testing the hypotheses \textit{$H_0^{(i)} : f$ is $G_i$-invariant} for certain subgroups $G_i \leq \mathcal{G}$ in the chosen lattice, and borrowing strength from these tests to make inferences about other subgroups. We also present two possible nonparametric tests to use in the nonparametric regression setting and the properties associated with these tests. The tests we present require mild technical assumptions and trade theoretical guarantees, computational time, and statistical power. Once a symmetry has been estimated we can then use it to adaptively smooth our estimate $\hat{f}$ non linearly, even reducing the dimension of the regression problem. \\
\\
The remainder of the paper is organised as follows. Section \ref{sec:background} covers essential background material and definitions (all in \textbf{bold face text}) particularly for groups and $G$-invariant functions. We then describe in section \ref{sec:inf_framework} how to structure the subgroups of the search group $\mathcal{G}$ into a lattice, describe confidence regions in this lattice, and give rules for how invariance behaves over this lattice. In section  \ref{sec:lattice} we then describe the estimation algorithm that searches for $G_{\max}(f, \mathcal{G} )$ by testing the invariance of $f$ to certain nodes of the lattice, and shows how we can borrow information across tests to improve the estimation. Section \ref{sec:testing} introduces two tests one can use in the estimation algorithm, each requiring different assumptions for the statistical problem. We then demonstrate the performance of the tests used, the symmetry estimator, and its use in nonparametric regression in simulation in section \ref{sec:exp}. We then apply these methods to satellite measurements of the earth's magnetic field in section \ref{sec:app}. Lastly, we conclude with a discussion of these methods in section \ref{sec:con}. All proofs can be found in the appendices, along with additional simulated and synthetic numerical results.

\section{Background}
\label{sec:background}

\subsection{Data Generating Mechanism}
\label{ssec:stats}
 Let $(\mathcal{X}, d_\mathcal{X})$ be a metric space and suppose that $\mu_X$ is a known Borel probability measure on $\mathcal{X}$. Suppose that we collect a data set $\mathcal{D} = \{ (X_i, Y_i) \}_{i =1}^n$ in $(\mathcal{X} \times \RR)^n$ where $X_i \iid \mu_X$ and with $\EE (Y_i \mid X_i = x) = f(x) \in L^2( \mathcal{X} )$ and $\EE( (Y_i - f(X_i) )^2 ) < \sigma^2 < \infty$ for all $i \in \NN$. We also assume that $\epsilon_i = Y_i - f(X_i)$ are independent for all $i$, so that the training data is independent and identically distributed (iid). In this paper we estimate the symmetries of the conditional mean function $f$, which we usually assume is a continuous \textbf{regression function} in some nonparametric function class $\mathcal{F} \subseteq L^2(\mathcal{X})$. This nonparametric regression setting will be the focus of this paper. We do not focus explicitly on the estimation of $f$, though the symmetries estimated can be used for this purpose with data augmentation, data projection, or feature averaging. \\
\\ 
  We write $\mathcal{D}_X = \{ X_i \}_{i = 1}^n$, with analogues for other random variables. We write $\hat{\mu}_X$ for the empirical distribution of $\mathcal{D}_X$. We will often consider functions classes $\mathcal{F}$ that have bounded variation, for example, the $(\beta, L)$-H\"{o}lder class on $\RR^d$ of $m = \lceil \beta \rceil - 1 $ times differentiable functions with $|f^{(m)}(x) - f^{(m)}(y)| \leq L d_\mathcal{X}( x, y )^{ \beta - m }$. These assumptions are fairly broad, as if $\mathcal{X}$ is compact then the Stone-Weierstrass theorem means that the space of continuous square integrable functions can be uniformly approximated by polynomials which are within the $(1, L)$-H\"{o}lder class for some $L$. We sometimes use $f$ as an index for $\PP_f$ and $\EE_f$ denoting the probability and expectation when the dataset is generated with a specific regression function $f$. 

\subsection{Group Theory}
\label{ssec:group_theory}

The symmetries that we focus on in this paper are best described mathematically with the language of group theory. In this section we present a quick introduction to the essential concepts needed for this methodology, and further details can be found in \citet{artin2018algebra} or any other abstract algebra textbook.  \\
\\
A \textbf{group} is set $G$ and an associative binary operation $(g, h) \mapsto gh \in G$ such that $G$ contains an \textbf{identity} $e \in G$ (such that $eg = ge = g$) and \textbf{inverses} $g^{-1}$ for each $g \in G$ such that $g^{-1} g = g g^{-1} = e$. We say that a subset $H$ of a group $G$ is a \textbf{subgroup} of $G$ if it is itself a group with the same binary operation, and write $H \leq G$. If $A \subseteq G$ then we can construct a subgroup of $G$ containing $A$ by adding all multiplications and inverses, or equivalently by taking the intersection of all subgroups containing $A$. We say that this subgroup is \textbf{generated} by $A$ and is written $\langle A \rangle$, or $\langle a_1, a_2, \dots \rangle$ when $A = \{ a_1, a_2, \dots \}$. The binary operation is associative but not necessarily commutative (i.e., $gh$ is frequently distinct from $hg$). The \textbf{order} of an element $g \in G$ is the smallest positive integer $a$ such that $g^a = e$. The \textbf{order} of a group is the number of elements, which can be infinite. Groups describe many objects considered in statistics; for example additive models for diffusion tensors (symmetric positive definite matrices) are built from their group structure in  \citet{lin2020additive}. \\
\\
A group encodes the structure of the symmetries of a regression function $f : \mathcal{X} \rightarrow \RR$ through a \textbf{group action}, a function $\cdot : G \times \mathcal{X} \rightarrow \mathcal{X}$ that obeys the rules $e \cdot x = x$ and $g \cdot ( h \cdot x ) = (gh) \cdot x$ for all $g, h \in G$ and $x \in \mathcal{X}$. These are the transformations in $\Phi$ described in the introduction. We say that an action is \textbf{faithful} if $g \cdot x = x$ for all $x$ only if $g = e$ (or equivalently, every non identity $g \in G$ has some point $x \in \mathcal{X}$ such that $g \cdot x \neq x$). Note that if $G$ acts faithfully then any subgroup of $G$ acts faithfully too. This notion is required for the identifiability of the subgroups we wish to estimate, if the action is not faithful then there are multiple subgroups that we cannot distinguish from the data. If $G$ acts on $\mathcal{X}$ and $H$ is a subgroup of $G$ then $H$ inherits an action on $\mathcal{X}$ from $G$ (i.e., $h \cdot_H x = h \cdot_G x$ as $h \in H \subseteq G$).  \\
\\
Let $G$ be a group that acts on $\mathcal{X}$. For any point $x \in \mathcal{X}$, the symmetries in $G$ map $x$ to a set of points called an \textbf{orbit}, written $[x]_G = \{ g \cdot x : g \in G \}$.  We also sometimes write $[n] = \{1, \dots, n \}$ and so keep the subscript $G$ for clarity. The set of orbits is called the \textbf{quotient space} of the action and is written $\mathcal{X} / G = \{ [x]_G : x \in \mathcal{X} \}$. These quotient spaces have a rich structure derived from the properties of the group, but it is only really used here when describing methods for estimating the function $f$ using the information of a symmetry $G$, so we omit further details here. Kendall's shape spaces \citep{dryden2016statistical} form an important example of these quotient spaces, where the orbits are the possible orientations, translations, and scalings of a set of landmark points.   \\
\\
%
Lastly we require notions of random variables that take values on groups. To do this, we introduce a topology to the group, so that we can construct a Borel probability space on it. We require that the topology agrees with the group operation in the sense that $(g, h) \mapsto g^{-1} h$ is continuous. We frequently use topological adjectives (e.g. compact, second countable) when referring to topological groups. We also require that the action of the group is continuous, and therefore measurable. If $G$ is compact then there exists a left (and right) shift invariant Haar measure on $G$ \citep{haar1933massbegriff} which we denote $\Gamma_G$ (or just $\Gamma$ if $G$ is clear from context). This means that $\Gamma( \{ ga : a \in A \} ) = \Gamma (A)$ for all open sets $A \subseteq G$ and $g \in G$. This measure can then be normalised to form an analogue of the uniform distribution on $G$, also sometimes written $U(G)$. In the case of directional statistics \citep{mardia2009directional}, the direction space $S^1 = \{ z \in \CC : |z| = 1 \}$ with the operation of complex multiplication forms a topological group, and the Haar measure precisely agrees the uniform distribution in the usual sense.

\subsubsection{Important Groups}
\label{sssec:imp_groups}

We now consider several important groups that will be used in the examples in this paper. All groups considered are second-countable and Hausdorff, and are almost always \textbf{matrix groups}, i.e.,  subgroups of the \textbf{ general linear group } $GL( \mathcal{X} ) = \{ T : T \text{ is an invertible linear operator } \mathcal{X} \rightarrow \mathcal{X} \}$ when $\mathcal{X}$ is a linear space such as $\RR^d$. These groups all act by applying the transformation to the points unless specified otherwise. When $\mathcal{X}$ is a $d$-dimensional vector space over $\RR$ we sometimes write $GL(d)$, and identify the linear operators in $GL(d)$ with their matrix representations $A = (T e_1, Te_2, \dots, T e_d )$ in some chosen basis $\{ e_i \}_{i = 1}^d$.  \\
\\
The \textbf{special linear group} written $SL(d)$ is the subgroup of $GL(d)$ consisting of all $A$ with $\det ( A) = 1$. The \textbf{ orthogonal group } written $O(d)$ consists of all $A$ such that $A^T A = I$. The \textbf{special orthogonal group} written $SO(d)$ consists of the intersection of the special linear and orthogonal groups on $\mathcal{X}$, and is the prototypical example of a compact group in this paper. It has appeared in statistics in nearly all of the cited work so far \citep{kim1998deconvolution, mardia2009directional, kendall1984shape}.  The group $SO(\RR^2)$ is homeomorphic (i.e., topologically equivalent) to the complex unit circle $S^1 = \{ z \in \CC : |z| = 1 \}$, and in fact this circle forms a group under complex multiplication that is also algebraically equivalent\footnote{i.e., isomorphic, which means that there is a bijection $\psi: SO(\RR^2) \rightarrow S^1$ with $\psi(R_1 R_2) = \psi(R_1) \phi(R_2)$ for all $R_1, R_2 \in SO(\RR^2)$.} to $SO(\RR^2)$. We frequently use these notations interchangeably. There are infinitely many subgroups of $SO(\RR^3)$ that are equivalent to $S^1$, corresponding to rotations of the hyperplane orthogonal to a unit vector $u$. To distinguish these as subgroups of $SO(\RR^3)$ we write $S^1_u$.    \\
\\
Another important set of transformations are those that preserve specific structures. The first example is the \textbf{dihedral group} of the square embedded in $\RR^2$, given by $D_4 =  \{ I,  R_{\pi/2}, R_{\pi}, R_{3 \pi / 2}, R_h , R_v, R_/, R_\backslash \} $ where:
	\begin{align*}
		I &= \begin{bmatrix} 1 & 0 \\ 0 & 1 \end{bmatrix}	&
		R_{\pi/2} &= \begin{bmatrix} 0 & - 1 \\ 1 & 0 \end{bmatrix}	&
		R_{\pi} &= \begin{bmatrix} -1 & 0 \\ 0 & -1 \end{bmatrix} &
		R_{3 \pi / 2} &= \begin{bmatrix} 0 & 1 \\ -1 & 0 \end{bmatrix} \\
		R_h &= \begin{bmatrix} -1 & 0 \\ 0 & 1 \end{bmatrix}	&
		R_v &= \begin{bmatrix} 1 & 0 \\ 0 & -1 \end{bmatrix}	&
		R_/ &= \begin{bmatrix} 0 & 1 \\ 1 & 0 \end{bmatrix} &
		R_\backslash &= \begin{bmatrix} 0 & -1 \\ -1 & 0 \end{bmatrix}	
	\end{align*}

This group of matrices has the property that if $A = \{ a_1, a_2, a_3, a_4 \} \subseteq \RR^2$ form the vertices of a square centred at the origin then $g \cdot A = \{ g \cdot a_i : a_i \in A \} = A$, thus preserving the square. This group is a subgroup of the orthogonal group on $\RR^2$, and has clear actions in image analysis.

\begin{eg}
	Consider the MNIST data set of handwritten digits \cite{lecun1998gradient}. The symmetries of $D_4$ act on these images, for example $R_h$ would take an image of a $3$ to an $\mathcal{E}$, but would preserve an $8$. In contrast $R_\backslash$ would take $3$ to an $\omega$, and $8$ to $\infty$, but would leave $O$ as $O$ (perhaps depending on the roundness of the $O$).  
\end{eg}

\subsection{$G$-invariant Functions}
\label{ssec:invs}

Let $f \in L^2(\mathcal{X}, \mu_X)$ be a function and let $G$ be any group acting measurably on $\mathcal{X}$. We say that $f$ is \textbf{strongly $G$-invariant} if $f(g \cdot x) = f( x) $ for all $g \in G$ and $x \in \mathcal{X}$. If $\| f - f_0 \|_2 = 0$ for some strongly $G$-invariant function $f_0$, we say that $f$ is \textbf{$G$-invariant}. Since we will primarily be restricting ourselves to the case where $f$ is continuous (as with the H\"{o}lder classes above) and when the support of $\mu_X$ is dense in $\mathcal{X}$, we can usually just consider the strongly $G$-invariant function in the $L^2$ equivalence class. If the action of $G$ and the function $f$ are both continuous and $G$ and $\mathcal{X}$ are separable spaces then $G$-invariance can be characterised by $\PP( f(g \cdot X) = f(X) ) = 1$ where $X \sim \mu_X$ and $g \sim \Gamma_G$.  \\
\\
 We say that a measure $\mu$ on $\mathcal{X}$ is \textbf{$G$-invariant} if $\mu( g \cdot A ) = \mu(A)$ for all $g \in G$ and measurable $A$, noting that the measurability of the action guarantees that $g \cdot A$ is measurable. The $G$-invariant square integrable functions form a linear subspace of $L^2(\mathcal{X}, \mu_X)$ (as invariance is trivially preserved by point wise addition and scalar multiplication). In \cite{elesedy2021provablyLin} they show that if $G$ is compact, Hausdorff, and second countable then this subspace is closed and $S_G$ is an orthogonal projection operator.


\section{Inference for Subgroups}
\label{sec:inf_framework}

We now describe the structure used to do inference on the subgroups of some group $\mathcal{G}$. The first step is to structure the subgroups into a lattice. This is a well known mathematical result, see \citet{schmidt2011subgroup,sankappanavar1981course}. We then show that the set of closed subgroups of $\mathcal{G}$ also form a lattice, which is not always a sub-lattice of the lattice of all subgroups. With this structure we are then able to ask the first statistical questions: how can one do estimation and inference on a lattice of subgroups? We also show that we can reduce the lattice to the information stored in a particular set of subgroups, which is used to minimise the number of tests conducted for both statistical and computational benefit.

\subsection{Lattices}
\label{ssec:gen_lattices}

A \textbf{lattice} $L$ is a partially ordered set with unique supremums and infimums. Given two elements $A, B \in L$ the unique supremum is called their \textbf{join} and is written $A \vee B$, and their unique infimum is written $A \wedge B$. Joins and meets are associative, commutative, idempotent binary operations that satisfy the \textbf{absorption laws} $a \wedge (a \vee b) = a$ and $a \vee (a \wedge b) = a$.  We say that a subset of a lattice $L$ is a \textbf{sub-lattice} of $L$ if it is closed under the joins and meets of the original lattice. Given some set $A \subseteq L$, we can take $\llangle A \rrangle$ as the smallest sub-lattice of $L$ that contains $A$ and say that $A$ \textbf{generates} $\llangle A \rrangle$. We can arrange a lattice $L$ into a \textbf{Hasse diagram}, a directed graph with nodes given by the elements $A, B \in L$ and edges from $A$ to $B$ if $B \leq A$, and there is no $C \in L$ with $B < C < A$. We typically omit the arrows on the edges in this paper as they all point down. 

\begin{eg}
In the context of variable selection we have a lattice of subsets of the variable indices, where the index subsets $A, B \subseteq \{ 1 , \dots, d \}$ have joins $A \vee B = A \cup B$ and meets $A \wedge B = A \cap B$.	
\end{eg}

\subsection{Subgroup Lattices}
\label{ssec:lattices}
The subgroups of a group $\mathcal{G}$ form a lattice with the order given by subgrouping \citep{schmidt2011subgroup,sankappanavar1981course}. We write $L(\mathcal{G})$ for this lattice of subgroups. In this lattice the join is given by the groups generated by unions, i.e., $A \vee B = \langle A, B \rangle$ and meet by intersections, i.e., $A \wedge B = A \cap B$, for $A, B \in L$.  Note that in the group context we can't simply take unions as these are not guaranteed to be groups, but the idea is essentially taking a union and adding in all the elements needed to form a group.

\begin{eg}
\label{eg:dihedral}
Consider the \textbf{dihedral group} of symmetries of the square $D_4$ (see section \ref{sssec:imp_groups}). This group has ten subgroups, including itself and the trivial group $I$. These are shown in the table of figure \ref{fig:hasse_D4} These can be ordered into the lattice in Figure \ref{fig:hasse_D4}, represented as a Hasse diagram.

	\begin{figure}[h]
	\centering
	\small
	
	\begin{tabular}{p{5.2cm} m{0.4cm} m{7cm} }
		
	\begin{tabular}{p{1.5cm} p{3cm}} 
	Subgroup  & Elements \\ \hline \hline
	$D_4$ & $\{ I, R_{\pi / 2}, R_{\pi}, R_{3\pi / 2}, $ \\
		  & $\phantom{\{} R_h, R_v, R_/, R_\backslash \}$ \\
	$\langle R_h, R_\pi \rangle$ & $\{ I, R_h, R_v, R_\pi \}$ \\
	$\langle R_{\pi/2} \rangle$ & $\{ I, R_{\pi / 2}, R_{\pi}, R_{3 \pi / 2 } \}$ \\
	$\langle R_h, R_\pi \rangle$ & $\{ I, R_h, R_v, R_\pi \}$ \\
	$\langle R_h \rangle$ & $\{ I, R_h \}$ \\
	$\langle R_v \rangle$ & $\{ I, R_v \}$ \\
	$\langle R_\pi \rangle$ & $\{ I, R_\pi \}$ \\
	$\langle R_/ \rangle$ & $\{ I, R_/ \}$ \\
	$\langle R_\backslash \rangle$ & $\{ I, R_\backslash \}$ \\
	$I$ & $\{ I \}$
	\end{tabular}
	
	&
	
	&
	\begin{tikzpicture}[scale = 0.7]
		\draw[fill] (0,0) circle (0.07cm) node[below]{$I$};
		\draw[fill] (-4,2) circle (0.07cm) node[below left]{$\langle R_h \rangle  $};
		\draw[fill] (-2,2) circle (0.07cm) node[below left]{$\langle R_v \rangle$};
		\draw[fill] (0,2) circle (0.07cm) node[below right]{$\langle R_{\pi} \rangle $};
		\draw[fill] (2,2) circle (0.07cm) node[below right]{$ \langle R_{/} \rangle  $};
		\draw[fill] (4,2) circle (0.07cm) node[below right]{$\langle R_{\backslash} \rangle  $};
		\draw[fill] (-3,4) circle (0.07cm) node[above left]{$\langle R_h, R_\pi \rangle$};
		\draw[fill] (0,4) circle (0.07cm) node[above right]{$\langle R_{\pi/2} \rangle $};
		\draw[fill] (3,4) circle (0.07cm) node[above right]{$\langle R_/, R_\pi \rangle$};
		\draw[fill] (0,6) circle (0.07cm) node[above]{$ D_4 $};
		\draw (0,0) -- (-4,2) -- (-3,4);
		\draw (0,0) -- (-2,2) -- (-3,4) -- (0,6);
		\draw (0,0) -- (0,2) -- (-3,4);
		\draw (0,0) -- (2,2) -- (3,4) -- (0,6);
		\draw (0,0) -- (4,2) -- (3,4);
		\draw (0,2) -- (0,4) -- (0,6);
		\draw (0,2) -- (3,4);
	\end{tikzpicture}
	
	\end{tabular}
	
	\caption{The list of subgroups and their elements of the dihedral group $D_4$ (see section \ref{sssec:imp_groups}), and the Hasse diagram of the subgroup lattice. Each dot represents a subgroup, with edge from $A$ down to $B$ if $A \leq B$ and there is no other group that can be ordered between them.}
	\label{fig:hasse_D4}
	\end{figure}
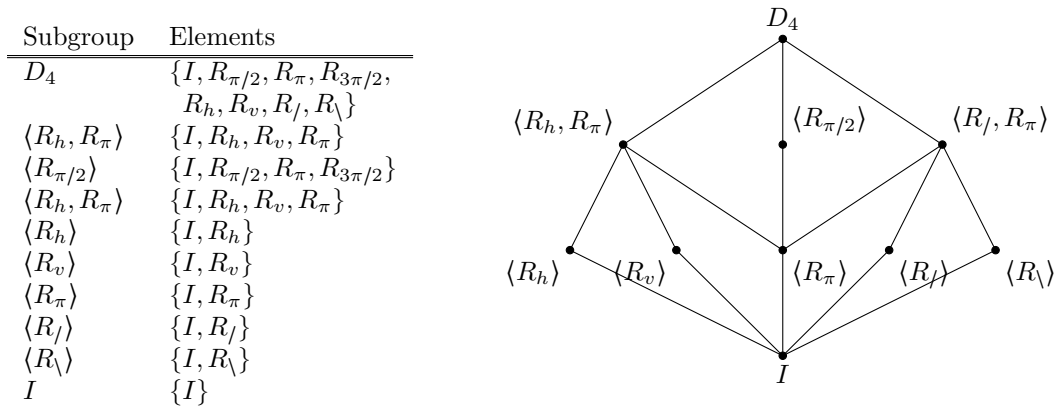
\end{eg}

One important subset of $L(\mathcal{G})$ is the collection of \textit{closed} subgroups (closed in the topology of $\mathcal{G}$), which we write $K(\mathcal{G})$ here. These do not necessarily form a sub-lattice of $L(\mathcal{G})$ for arbitrary $\mathcal{G}$ as we can generate non closed subgroups with the original join operator $\langle G, H \rangle$ (consider two copies of high order finite rotations around distinct axes generating a proper dense subgroup of $SO(3)$). It does form a lattice as it still has unique supremums (just distinct ones from the supremum in $L$) and infimums, where the join is the topological closure of the group generated from $G$ and $H$, i.e., $\overline{ \langle G, H \rangle}$. This is established in the following proposition. In some cases (e.g. finite $\mathcal{G}$) this is equivalent to the original join, so the lattice in example \ref{eg:dihedral} is both $L(D_4)$ and $K(D_4)$. 

\begin{prop}
	\label{prop:closed_sublattice}
	Let $\mathcal{G}$ be a topological group. The closed subgroups $K(\mathcal{G})$ form a lattice with partial ordering given by subgrouping, meets given by intersection, and joins given by the topological closure of the group generated by the two subgroups.
\end{prop}


\begin{eg}
\label{eg:circle}
Consider the circle group $S^1 = \{ z \in \CC : |z| = 1 \}$. The finite subgroups are those of rational rotations, and are of the form $ C_a = \langle \exp( 2 \pi i / a ) \rangle $ of order $a$. These have subgroups corresponding to the factors of $a$, and so the lowest parts of the subgroup lattice of $S^1$ are the groups of prime order. Above these you have groups with two (counting with repetition) prime factors, and so on. There are also subgroups of countable infinite order, for example $\langle \exp( i ) \rangle$, and one subgroup of uncountable order - itself. Importantly, all of these infinite subgroups are all dense in $S^1$ (under the subset topology from $\CC$). Thus the key difference between $L(S^1)$ and $K(S^1)$ is that these dense subgroups are missing.
\end{eg}

\subsection{Confidence Regions in Lattices}

As statistical objects, groups typically have relationships characterised by inclusions - if a symmetry $G$ applies in a particular situation then any subgroup $H \leq G$ surely must also apply. Therefore we are typically interested in confidence sets that contain all of their subgroups. We formalise this below, and give a measure of the optimality of these confidence sets.  \\


Suppose we wish to estimate a group object $G_0$ within some lattice $L$, and that the lattice has a measure $\nu$. We say that a region $\hat{R} \subseteq L$ is a $(1-\alpha)$-\textbf{confidence region} for $G_0$ if it is a measurable function of the data, contains all subgroups of all of its elements, and 
\begin{equation}
	\PP_f( G_0 \in \hat{R} ) \geq 1 - \alpha.
\end{equation}
The ideal confidence region would be the set of subgroups of $G$, which we denote $A$, so we judge the optimality of the region $\hat{R}$ by the measure of the groups that are different to $A$:
\begin{align}
	\mathcal{E}( \hat{R} ) = \nu( \{ G \in \hat{R} : G \notin A \} \cup \{ G \in A : G \notin \hat{R} \} ) = \nu\big( ( \hat{R} \cup A ) \setminus ( \hat{R} \cap A ) \big)
\end{align}
where we assume that each of these is a measurable set. In our method latter each lattice $L$ will be finite and $\nu$ will be the discrete counting measure, so this will be true. 

\begin{eg}
Consider the group $D_4$ and its subgroup $G_0 = \langle R_h, R_{\pi} \rangle$. In figure \ref{fig:d4_cr_eg}, we show three examples and non-examples of possible samples of confidence regions for $G_0$, and show the excess over $A$. 

\begin{figure}[h]
	\centering
	\small
	\begin{tabular}{m{4.5cm} m{4.5cm} m{4.5cm}}
	\centering	
	\begin{tikzpicture}[scale = 0.5]
		\draw[fill, blue] (0,0) circle (0.07cm) node[below]{$I$};
		\draw[fill, blue] (-4,2) circle (0.07cm);
		\draw[fill, blue] (-2,2) circle (0.07cm);
		\draw[fill, blue] (0,2) circle (0.07cm);
		\draw[fill] (2,2) circle (0.07cm);
		\draw[fill] (4,2) circle (0.07cm);
		\draw[fill, blue] (-3,4) circle (0.07cm) node[above]{$\langle R_h, R_\pi \rangle$};
		\draw[fill] (0,4) circle (0.07cm);
		\draw[fill] (3,4) circle (0.07cm);
		\draw[fill] (0,6) circle (0.07cm) node[above]{$ D_4 $};
		\draw[blue] (0,0) -- (-4,2) -- (-3,4);
		\draw[blue] (0,0) -- (-2,2) -- (-3,4);
		\draw (-3,4) -- (0,6);
		\draw[blue] (0,0) -- (0,2) -- (-3,4);
		\draw (0,0) -- (2,2) -- (3,4) -- (0,6);
		\draw (0,0) -- (4,2) -- (3,4);
		\draw (0,2) -- (0,4) -- (0,6);
		\draw (0,2) -- (3,4);
	\end{tikzpicture}
	&
	\centering
	\begin{tikzpicture}[scale = 0.5]
		\draw[fill, blue] (0,0) circle (0.07cm) node[below]{$I$};
		\draw[fill, blue] (-4,2) circle (0.07cm);
		\draw[fill, blue] (-2,2) circle (0.07cm);
		\draw[fill, blue] (0,2) circle (0.07cm);
		\draw[fill] (2,2) circle (0.07cm);
		\draw[fill] (4,2) circle (0.07cm);
		\draw[fill, blue] (-3,4) circle (0.07cm) node[above]{$\langle R_h, R_\pi \rangle$};
		\draw[fill, red] (0,4) circle (0.07cm);
		\draw[fill] (3,4) circle (0.07cm);
		\draw[fill] (0,6) circle (0.07cm) node[above]{$ D_4 $};
		\draw[blue] (0,0) -- (-4,2) -- (-3,4);
		\draw[blue] (0,0) -- (-2,2) -- (-3,4);
		\draw (-3,4) -- (0,6);
		\draw[blue] (0,0) -- (0,2) -- (-3,4);
		\draw (0,0) -- (2,2) -- (3,4) -- (0,6);
		\draw (0,0) -- (4,2) -- (3,4);
		\draw[red] (0,2) -- (0,4);
		\draw (0,4) -- (0,6);
		\draw (0,2) -- (3,4);
	\end{tikzpicture}
	&
	\begin{tikzpicture}[scale = 0.5]
		\draw[fill, blue] (0,0) circle (0.07cm) node[below]{$I$};
		\draw[fill, blue] (-4,2) circle (0.07cm);
		\draw[fill, blue] (-2,2) circle (0.07cm);
		\draw[fill, blue] (0,2) circle (0.07cm);
		\draw[fill] (2,2) circle (0.07cm);
		\draw[fill] (4,2) circle (0.07cm);
		\draw[fill, blue] (-3,4) circle (0.07cm) node[above]{$\langle R_h, R_\pi \rangle$};
		\draw[fill, red] (0,4) circle (0.07cm);
		\draw[fill] (3,4) circle (0.07cm);
		\draw[fill, red] (0,6) circle (0.07cm) node[above]{$ D_4 $};
		\draw[blue] (0,0) -- (-4,2) -- (-3,4);
		\draw[blue] (0,0) -- (-2,2) -- (-3,4);
		\draw[red] (-3,4) -- (0,6);
		\draw[blue] (0,0) -- (0,2) -- (-3,4);
		\draw (0,0) -- (2,2) -- (3,4) -- (0,6);
		\draw (0,0) -- (4,2) -- (3,4);
		\draw[red] (0,2) -- (0,4);
		\draw[red] (0,4) -- (0,6);
		\draw (0,2) -- (3,4);
	\end{tikzpicture} \\
	\centering (a) $\mathcal{E}( \hat{R}_1 ) = 0$ &
	\centering (b) $\mathcal{E}( \hat{R}_2 ) = 1$ &
	\centering (c) Not a confidence region 
	\end{tabular}
	
	\caption{Two samples of confidence region for the subgroup $G_0 =  \langle R_h, R_{\pi} \rangle$ of $D_4$ shown in blue with excess in red. Also shown is an example of a subset that is not a confidence region, which fails because it does not include all of the subgroups of its elements. The errors when $\nu$ is the counting measure are given underneath each picture.}
	\label{fig:d4_cr_eg}
\end{figure}
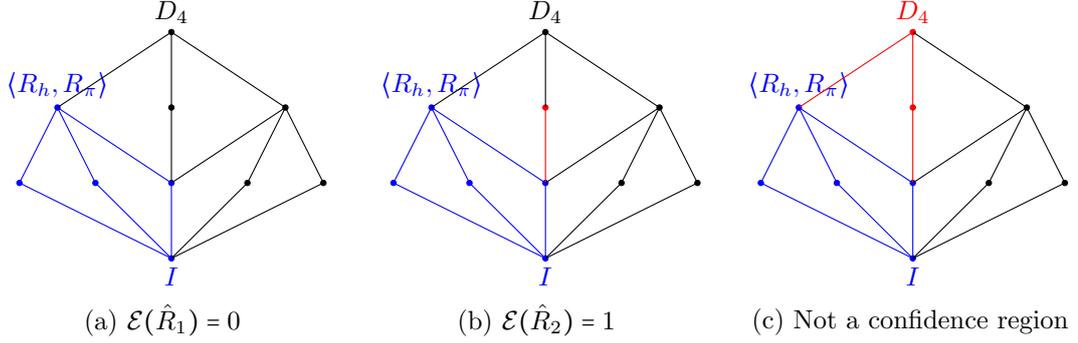
\end{eg}

\subsection{Maximal Subgroups in Regression}

We now define the object that we seek to estimate in section \ref{sec:lattice}. Suppose $f: \mathcal{X} \rightarrow \RR$ is a regression function and $\mathcal{G}$ acts measurably on $\mathcal{X}$. We say that a subgroup $G \leq \mathcal{G}$ is a \textbf{maximal invariant subgroup} for $f$ if: (1) $f$ is $G$-invariant; and (2) if $f$ is $H$ invariant then $H \leq G$. \\
\\
To prove the existence and uniqueness of such maximal invariant subgroups, we first characterise how invariance behaves over a lattice. In essence, invariance is preserved by taking joins and meets. It is also preserved by the meet of an invariant group and a non-invariant group. 

\begin{prop}[Rules for Invariance]
	\label{prop:rules}
	Suppose that $\mathcal{G}$ is a group acting measurably on $\mathcal{X}$, and let $f : \mathcal{X} \rightarrow \RR$ be a regression function. Suppose that $G, G'$ and $H$ are subgroups of $\mathcal{G}$. Then:
	\begin{enumerate}
	\item If $f$ is $G$-invariant, then $f$ is $H$-invariant for any $H \leq G$;
	\item If $f$ is not $G$-invariant , the $f$ is not $H$-invariant for any $H \geq G$;
	\item If $f$ is both $G$- and $H$-invariant, then $f$ is $\langle G, H \rangle$-invariant; and
	\item If $f$ is $G$-invariant but not $H$-invariant, then $f$ is not $G'$ invariant for any $G' \leq \mathcal{G}$ where $\langle G, H \rangle = \langle G, G' \rangle$.
\end{enumerate}
Moreover, if $\mathcal{G}$ is a closed group acting continuously and $f$ is continuous then the groups $\langle G, H \rangle$ and $\langle G, G' \rangle$ can be replaced by their topological closures above.
\end{prop}

Using these rules we can prove the following result that guarantees the existence and uniqueness of a maximal invariant subgroup, and we can utilise them in the estimation algorithm in section \ref{sec:lattice}.

\begin{prop}
\label{prop:max_subgroup}
For any group $\mathcal{G}$ acting on $\mathcal{X}$, every function $f : \mathcal{X} \rightarrow \RR$ has a unique maximal invariant subgroup $G_{\max}(f, \mathcal{G})$. If $f$ is continuous and $\mathcal{G}$ acts continuously on $\mathcal{X}$ then $G_{\max}(f, \mathcal{G})$ is closed.  
\end{prop}

Unlike the variable selection problem, the lattice $K( \mathcal{G} )$ is often infinite both breadth-wise and depth-wise. Therefore it is important for us to consider finite portions of this lattice in our estimation method. Given a chosen sub-lattice $K$ of $K(\mathcal{G})$, the following proposition shows that there is a well defined maximal object that we can look to estimate from the data.

\begin{prop}
	\label{prop:sub_lattice_restriction}
	Suppose that $\mathcal{G}$ acts faithfully and continuously on $\mathcal{X}$, and that $f$ is continuous. Let $K$ be a finite sub-lattice of $K(\mathcal{G})$ that contains the trivial group $I$. Then there exists a unique node $G_{\max}(f, K )$ in $K$ that $f$ is invariant to, and for which $H \leq G_{\max}(f, K )$ for all $H \in K$ where $f$ is $H$-invariant.
\end{prop}

\subsection{Lattice Bases}
\label{ssec:lattice_bases}
Lastly we require a structure theorem for lattices, giving a condition on a set of subgroups that simplifies the structure of the lattice we will work over. Importantly, it allows us to extrapolate from $\ell$ tests to a lattice of size up to $2^\ell$. This is a huge benefit of considering the group structure of the transformations. \\
\\
 We say that a set of closed subgroups $\{ G_i \}_{i = 1}^\ell$ of $\mathcal{G}$ satisfies the \textbf{intersection property} (I) if $G_i \cap G_j \in \{ I, G_i, G_j \}$ whenever $i \neq j$. Any set of subgroups satisfying the intersection property will be called a \textbf{lattice base}.

\begin{prop}	
	\label{prop:lattice_as_vees}
	 Suppose the set of subgroups $\{ G_i \}_{i = 1}^\ell$ of $\mathcal{G}$ satisfies the intersection property (I). Then the smallest lattice (in the sub lattice ordering) containing $G_1, \dots, G_\ell$ consists only of groups of the form $G_{i_1} \vee \cdots \vee G_{i_m}$ for some distinct $i_j \in [\ell]$, as well as the empty product $I$. 
\end{prop}

This result allows us to identify the lattice $K = \llangle \{  G_i \}_{ i = 1 }^{\ell} \rrangle$ with subsets of $\{ 1, \dots \ell \}$, but where some subsets represent the same group in the lattice. 

\begin{eg}
	In the variable selection lattice of figure \ref{fig:demo_alg}a, a lattice base is the singletons set $\{ \{ i \} : i \leq p \}$. Equivalently for the translation group framing in figure \ref{fig:demo_alg}b, a lattice base is the collection $\{ \RR_i : i \leq p \}$. 
\end{eg}

\begin{eg}
	Consider the group $D_4$. One lattice base that generates the entire lattice $L(D_4)$ is given by:
	\begin{equation}
		B = \big\{ \langle R_h \rangle, \langle R_v \rangle, \langle R_{\pi} \rangle, \langle R_/ \rangle, \langle R_\backslash \rangle, \langle R_{\pi / 2} \rangle \big\} 
	\end{equation}
	It is easy to check that all but one on the pairwise intersections are trivial; the only non-trivial intersection being $\langle R_{\pi / 2} \rangle \cap \langle R_\pi \rangle = \langle R_\pi \rangle$.
\end{eg}


%
%

\section{Estimation of Maximal Invariant Subgroups}
\label{sec:lattice}

In the context of nonparametric regression data as per section \ref{ssec:stats}, we seek to estimate the unique maximal invariant subgroup for $f$ of some search group $\mathcal{G}$ acting continuously and faithfully on the domain $\mathcal{X}$. We assume that $f$ is continuous, so that this maximal subgroup is closed. We will estimate $G_{\max}( f, \mathcal{G} )$ by considering a sequence of finite sub lattices $K_\ell$ that approximate $K(\mathcal{G})$. Each finite sub-lattice $K_\ell$ will be generated by lattice bases $\{ G_i \}_{i = 1}^\ell$ such that $\{ G_i \}_{i = 1}^\ell \subseteq \{ G_i \}_{i = 1}^{\ell + 1}$. Choices for this base are discussed in section \ref{ssec:choose_lattice}.  \\
\\
Suppose that we have a test for the hypothesis \textit{$H_0^{(i)}: f$ is $G_i$-invariant} against the alternative \textit{$H_1^{(i)} : f$ is not $G$-invariant} for each of the groups $G_i$ in the lattice base. We will provide two possibilities for such a test in the next section. We will test these hypotheses at significance levels $\alpha_i$, such that $\sum_{i \in \NN} \alpha_i = \alpha \in [0,1]$. The choice of $\alpha_i$ will depend on how many terms of the lattice base we wish to test, and could either have finitely many non-zero terms or infinitely many, as long as the sum converges. Choices for $\alpha_i$ are discussed in section \ref{ssec:choose_alpha_i}. \\
\\
We then use the information of these tests to infer whether $f$ is $G$-invariant for each $G \in K_\ell$ using the rules for invariance of Proposition \ref{prop:rules}. We will first describe the process when finitely many tests have been conducted, and provide theoretical and computational results for the algorithm to that point.

\subsection{Estimating Over a Finite Lattice}

Let $\{ G_i \}_{i = 1}^\ell$ be a lattice base, and consider the lattice $K_\ell = \llangle \{ G_i \}_{i = 1}^\ell \rrangle$, which by Proposition \ref{prop:lattice_as_vees} contains only groups of the form $G_{i_1} \vee \cdots \vee G_{i_m}$ for some distinct $i_j \in [\ell]$. The rules for invariance (Proposition \ref{prop:rules}) guarantee that knowledge of invariance of the lattice generators $G_i$ gives knowledge of invariance at all other groups $G$ in the lattice $K_\ell$. Therefore after we test each of the hypotheses $H_0^{(i)}$, we can then define:
\begin{equation}
	\hat{K}_{\ell} = \big\{ G \in K_\ell : \mathrm{Test}_{\alpha_i} ( \mathcal{D}, G_i ) = 1 \text{ for all } G_i \leq G \big\}.
\end{equation}
where $\mathrm{Test}_{\alpha_i}( \mathcal{D} , G_i )$ is the test for the null hypothesis $H_0^{(i)}$ at significance level $\alpha_i$, taking the value $0$ if this hypothesis is rejected and $1$ otherwise. This object can be computed (after testing) as long as we can compute the maps $\mathbf{1}_{\{ G_i \leq G\} }$ for each $G \in K_\ell$. We describe this computation in section \ref{ssec:comp_K}. 

\begin{prop}
	\label{prop:conf_region}
	The region $\hat{K}_{\ell}$ is a $(1 - \alpha)$-confidence region for $G_{\max}(f, K_\ell)$. In particular, for all $f \in \mathcal{F}$ we have:
	\begin{equation}
		\PP_f( G_{\max}(f, K_\ell) \in \hat{K}_{\ell} ) \geq 1 - \sum_{i = 1}^\ell \alpha_i \geq 1 - \alpha.
	\end{equation} 	
\end{prop}

From this confidence region $\hat{K}_{\ell}$, we can take a point estimate $\hat{G}_\ell$ as any maximal element. This will usually be unique, but if not then it can be chosen via any distribution on this set (usually the uniform distribution). 

\subsection{Computation of $\hat{K}_\ell$}
\label{ssec:comp_K}

In Algorithm \ref{algo:comp_K} we provide one possible method for computation of $\hat{K}_\ell$. The main requirement for this is the ability to compute the inclusions $G \leq H$ for $G, H \in K_\ell$. This algorithm has the advantage of being iterable; if we wish to compute $\hat{K}_{\ell + 1}$ then we simply store the $\mathrm{TestResults}$ vector as well as $\hat{K}_\ell$ and add another term to the for loop. \\

\begin{algorithm}[h]
\caption{Computation of $\hat{K}_\ell$}
\label{algo:comp_K}
\begin{algorithmic}[1]
\Procedure{ComputeKHat}{$\mathcal{D}$, $\{ G_i \}_{i = 1}^\ell$} 
	\State Initialise $\hat{K}_0 = \{ I \}$
	\State Initialise $\mathrm{Test Results} = ()$ as the empty tuple
	\For{ $i \in \{ 1, \dots \ell \}$ }
		\State Compute $\mathrm{Test}_{\alpha_i}( \mathcal{D}, G_i )$ and append to $\mathrm{Test Results}$
		\State Set $\hat{K}_i = \{ \}$. 
		\If{ $\mathrm{Test}_{\alpha_i}( \mathcal{D}, G_i ) = 1 $}
			\State Set $\hat{K}_i = \hat{K}_{ i - 1} \cup \{ G \vee G_i : G \in \hat{K}_{i - 1} \}$
			\For{ $j \in \{1, \dots, i - 1\}$ }
				\If{ $\mathrm{TestResults}_j = 0$ }
					\State Set $\hat{K}_i = \{ G \in \hat{K}_i : G_j \not\leq G \}$
				\EndIf
			\EndFor
		\Else
			\State Set $\hat{K}_i = \{ G \in \hat{K}_{i - 1} : G_i \not\leq G \}$
		\EndIf
	\EndFor
	\State Return $\hat{K}_\ell$
\EndProcedure
\end{algorithmic}
\end{algorithm}

The computational complexity of this algorithm will depend on the complexity of the $\ell$ tests conducted. For the tests we propose in this paper, see sections \ref{sssec:avt_comp} and \ref{sssec:pv_comp} for their complexities. A naive search for $G_{\max}$ would conduct up $2^\ell$ tests for every symmetry in $K_\ell$. Since the testing is usually the main computation, reducing to $O(\ell)$ tests saves significant computational time, and also allows for each test to be conducted at higher significance levels improving the statistical performance of the estimate.  \\
\\
Aside from this, we require (in the worst case) potentially $O(\ell^2  2^\ell)$ checks of inclusion $G \leq H$ (because $\hat{K}_\ell$ can be of size $O(2^\ell)$ and we might have to check inclusion $\ell^2$ times for each). This is entirely determined by the choice of lattice base and is independent of $n$ and the dimension of $\mathcal{X}$.

\subsection{Computation of $\hat{G}$}

Once $\hat{K}_\ell$ has been computed, we need to select a point estimate of $G_{\max}(f, K(\mathcal{G}))$. This amounts to filtering out the elements that are not maximal and then sampling uniformly from the remainder. \\

%

\begin{algorithm}[h]
\caption{Computation of $\hat{G}_\ell$}
\label{algo:comp_G}
\begin{algorithmic}[1]
\Procedure{ComputeGHat}{$\hat{K}_{\ell}$} 
	\State Set $\mathrm{MaxElems} = \{ I \}$
	\For{ $G \in \hat{K}_\ell$ } 
		\For{ $H \in \mathrm{MaxElems}$ }
			\If{ $G \geq H$ }
				\State Remove $H$ from $\mathrm{MaxElems}$
			\EndIf
			\If{ $G < H$ }
				\State Continue the $G$ for loop
			\EndIf
		\EndFor
		\State Add $G$ to $\mathrm{MaxElems}$
	\EndFor 
	\State Return $\hat{G}$ sampled from $\mathrm{MaxElems}$ uniformly
\EndProcedure
\end{algorithmic}
\end{algorithm}

Note that it is not a problem to remove $H$ from the maximal element set in line 6, because it is not maximal and this will not affect whether any maximal elements are added later. The worst case computational time can be bounded trivially by $| \hat{K}_\ell | \times | \max \hat{K}_\ell | \leq 2^{2\ell} $, though it is likely that the true worst case performance is faster because of the trade off in the sizes between $\hat{K}_\ell$ and $\max \hat{K}_\ell$.


\subsection{Choosing the Lattice Base}
\label{ssec:choose_lattice}

%

The choice of the lattice base $\{ G_i \}_{i = 1}^\ell$, is fundamental to the method. We want to cover large portions of $K(\mathcal{G})$ with the lattices $K_\ell$ for minimal numbers of tests conducted. Recall that the Hausdorff distance $d_{K(\mathcal{G})}$ between compact subgroups $G, H \leq \mathcal{G}$ is given by:
\begin{equation}
	d_{K(\mathcal{G})} ( G, H ) = \max \big( \sup_{g \in G} \inf_{h \in H} d( g, h) , \sup_{h \in H} \inf_{g \in G} d(h, g) \big)
\end{equation}
when $\mathcal{G}$ has a metric $d$, as with the matrix groups we have used as examples. 
The choice will depend on the goals of the lattice base selection. We have optimised for the following objectives: 
\begin{enumerate}[	(B1)]
	\item $K_\ell$ contains symmetries of interest to the practitioner; 
	\item Maximise the number of groups, $| K_\ell |$;
	\item Maximise the pairwise Hausdorff distances $d_{K(\mathcal{G})}(G_i,G_j)$ for the compact $G_i$;
	\item Minimise $|G_i|$ or $\dim G_i$ for all $G_i$ in the base; and
	\item Maximise $|G|$ or $\dim G$ for all $G \in K_\ell$.
\end{enumerate}
These goals are roughly in order of importance. If there are no symmetries of interest then the method described can still proceed, but of course any prior interest should be the top priority. The second goal speaks to the efficiency of this method: as the ratio $| K | / 2^\ell$ approaches 1 we get maximal information about the whole lattice $K(\mathcal{G})$ out of the $\ell$ tests conducted. The third spaces the groups around $K$, which ensures that you don't waste tests on groups that the data cannot distinguish. The fourth and fifth mean that we should run the tests on smaller groups where we can cover with samples more easily, but that they give information about larger groups. \\
\\
The following recursive method for selecting the lattice base provides a general solution that achieves a balance of the goals above. Suppose we are given a metric group $\mathcal{G}$, a desired number of groups $\ell$, and some starting symmetries of interest $B_0 = \{ G_1, \dots, G_k \} \subseteq K( \mathcal{G} )$ that form a lattice base. Then for $i = 1, 2, \dots$, we pick an element $g \notin \bigvee_{G \in B_i} G$ such that:
	\begin{enumerate}
		\item $\overline{ \langle g \rangle } \cap G \in \{ I, G, \overline{ \langle g \rangle} \}$ for all $G \in B_i$;
		\item If $\mathcal{G}$ is finite then $ | g | = \min_{h \notin \bigvee_{G \in B_i}} |h|$;
		\item $\dim \overline{ \langle g \rangle } \leq 1$; and
		\item $\prod_{G \in B_i} d_\mathcal{G}(g, G)$ is (roughly) maximised.
	\end{enumerate} 
	Then set $B_{i+1} = B_i \cup \{ \overline{ \langle g \rangle } \}$. The product in (iv) encourages equal spacing between the other groups for a given sum of the distances. This of course requires the practitioner to know about various aspects of the group theory of $\mathcal{G}$, which are well known for all of the groups in section \ref{sssec:imp_groups} among many others. The methods presented here are robust to the choice of the lattice, as demonstrated in example \ref{eg:est_f_sims} and appendix \ref{app:lattice_base_effects}. There we see that using the estimated symmetries gives predictive performance similar to or better than an unsymmetrised estimator, even when there is no symmetry present in the lattice chosen. 

%

\subsection{Choice of $\alpha_i$}

\label{ssec:choose_alpha_i}

There are a few simple choices for the significance levels $\alpha_i$ for a finite sample $n$. The first is to set $\alpha_i = \alpha / m$ for $i \leq m$, which is a Bonferroni correction for multiple testing. This is attractive in its simplicity but allows for only a finite portion of $K(\mathcal{G})$ to be searched. \\
\\
Another choice is to pick a sequence with positive values for all $i$, for example $\alpha = \tfrac{1}{2^i}$. This would allow for infinitely many groups to be tested, though with likely very low powers for the later groups.  \\
\\
Lastly, we can pick a fixed level $\alpha$ and simply test \textit{all} groups in the lattice base at this level, as is done in many practical variable selection problems \citep{heinze2018variable} (although not necessarily recommended given issues of post model selection inference). This of course can affect the probability $\PP_f( G_{\max}( f, K_\ell) \in \hat{K}_\ell)$ but in practice often seems not to increase it beyond $\alpha$ because the tests run on the same dataset are not independent. See appendix \ref{app:correct_containment} as a continuation of example \ref{eg:low_dim_G_sims} for an example of this.

\subsection{Theoretical Properties of $\hat{G}_\ell$ and $\hat{K}_\ell$}

\label{ssec:theory_g_hat}

We now turn to the understanding how the properties of the hypothesis test used determines the properties of the estimate $\hat{G}_\ell$ and the confidence region $\hat{K}_\ell$ for a fixed lattice base of size $\ell$. 


\begin{thm}
	\label{thm:cons_of_hat_G}
	Suppose that the hypothesis test used is consistent at every significance level $\alpha \in (0,1)$. For any fixed significance levels $\alpha_i$ (that do not depend on $n$), we have:
	\begin{enumerate}[	(1)]
		\item $\PP_f\big( \hat{G}_\ell \leq G_{\max}(f, K_\ell) \big) \rightarrow 1$ when $\hat{G}_\ell$ is sampled uniformly from $\max \hat{K}_\ell$.
	\end{enumerate}
	Moreover, there exist deterministic sequences of significance levels $\alpha_i(n)$ for which:
	\begin{enumerate}[	(1)]
		\setcounter{enumi}{1}
		\item $( \mathrm{Test}_{\alpha_i(n)}(\mathcal{D}, G_i ) : i  = 1, \dots, \ell ) \overset{\PP_f}{\longrightarrow} 1 - 2 ( \mathbf{1}_{ G_i \leq G_{\max}(f, K_\ell )} : i = 1, \dots, \ell )$;
		\item $\hat{G}_\ell$ is a consistent estimator of $G_{\max}( f, K_\ell )$, when $\hat{G}_\ell$ is sampled uniformly from $\max \hat{K}_\ell$;
		\item $\hat{K}_\ell \overset{\PP_f}{\longrightarrow} \{ G \in K_\ell : G \leq G_{\max}( f, K_\ell ) \}$; and
		\item $\mathcal{E}(\hat{K}_\ell)  \overset{\PP_f}{\longrightarrow} 0$,
	\end{enumerate}
	where $\overset{\PP_f}{\longrightarrow}$ denotes convergence in probability under the measure $\PP_f$. The sequence of significance values can be calculated if the power of $\mathrm{Test}_\alpha$ has a known lower bound.  
\end{thm}

The first statement says that we are probabilistically guaranteed eventually not to \textit{over}-estimate the symmetry. This means that the method presented is robust to type I errors of the tests used, as errors on the true symmetries of $f$ result result in missed symmetries but do not introduce errant ones. The statements (2)-(5) are equivalent formulations of the stronger statement of consistency.

\section{Hypothesis Testing for a Particular Symmetry}
\label{sec:testing}

We now turn to constructing tests for the hypothesis \textit{ $H_0 : f$ is $G$-invariant} against the alternative \textit{$H_1 : f$ is not $G$-invariant} for use in the subgroup estimation, where $G$ is any group acting measurably on $\mathcal{X}$. We present two options that  make trade-offs in terms of generality and guarantees. The first requires stronger assumptions but has a theoretical guarantee of convergence, whereas the second makes weaker assumptions but has fewer guarantees and requires more computation.

\subsection{Asymmetric Variation Test} 
\label{ssec:asym_var_test}
This test relies on the idea that if we know how quickly $f$ can vary locally (i.e., we have a relationship between $|f(x) - f(y)|$ and $d_\mathcal{X}( x, y)$) then we can construct a test statistic that measures the effect of permuting the covariates in $X_i$ under random transformations from the symmetry to test $G$. \\
\\
If we can assume that the regression function $f$ is in a class $\mathcal{F}$ of \textit{bounded variation}, then we can define $V(x,y) = \sup_{f \in \mathcal{F} }  |f(x) - f(y) |$. This bound $V$ is used to identify if $f$ varies more quickly than $V( g\cdot x, y)$. An example of such a class are $\beta$-H\"{o}lder continuous functions $\mathcal{F}(L, \beta)$ (defined in section \ref{ssec:stats}) for which $V(x,y) = L d_\mathcal{X}(x,y)^\beta$ when $\beta \in (0, 1]$. The only real choice here is $L$, as the continuous $f$ can be uniformly well approximated by Lipschitz functions of some steepness, so in practice one would use the bound $V(x,y) = L d_\mathcal{X} ( x, y)$ using an assumed $L$. If this bound is unknown then it could be estimated using techniques in appendix \ref{app:estimate_lipschitz_const}.  \\
\\
Suppose that we also have a bound $p_t$ such that $ \PP( | \epsilon_i - \epsilon_j | > t ) \leq p_t$ for the independent mean zero additive noise $\epsilon_i = Y_i - f(X_i)$, for all positive $t$. One example of this is $p_t = \frac{ 2 \sigma }{ t  } \tfrac{ \exp( - t^2 / 4 \sigma^2 )}{\sqrt{2\pi}}$ for independent and identically distributed Gaussian $\epsilon_i$ with variance $\sigma^2$ (Proposition 2.1 of \citet{adamsMA3K0notes}), or in general one can use a Chebyshev inequality. Let $\mu_g$ be any distribution on the group $G$, and let $g \sim \mu_g$.   \\
\\
We can now define the test statistic based on a resampling procedure. Let $m$ be the number of resamples. First, for $j \in \{ 1, \dots, m \}$ set $D_{I(j) J(j)}^{g_j} = |Y_{I(j)} - Y_{J(j)} | - V( g_j \cdot X_{I(j)}, X_{J(j)})$, where $g_j$ is sampled independently from $\mu_g$, $I(j)$ is sampled uniformly from $[n] = \{ 1, \dots, n \}$, and $J(j)$ is the index of the nearest neighbour to $g_j \cdot X_{I(j)}$ in the original data $\mathcal{D}_X = \{ X_i \}_{i = 1}^n$. Our test statistic is then $N_t^g = | \{ D_{I(j) J(j)}^{g_j} \geq t \} |$, the distribution of which is controlled by the following lemma.

\begin{Lem}
\label{lem:test_stat_control}
Under the null hypothesis of $G$-invariance, and for any fixed choice of resamples $m$, $N_t^g$ is stochastically bounded from above by a binomial random variable $Z$ with $m$ trials and probability of success $p_t$ for sufficiently large sample sizes $n$. I.e., 
\begin{equation}
	\PP( N_t^g > c ) \rightarrow p \leq \PP( Z > c ) 
\end{equation}
for all $c \in \RR$.  
\end{Lem}

Using this lemma it is straightforward to calculate $p$-values for this hypothesis via:
\begin{equation}
	p_{val} = \sum_{k = N_t^g}^{m} \binom{m}{k} p_t^k (1 - p_t)^{m - k}
\end{equation}
which are guaranteed to have asymptotic type I error control at any significance level $\alpha \in [0,1]$ . I.e., 
\begin{equation}
	\PP_f( p_{val} \leq \alpha ) \rightarrow a \leq \alpha
\end{equation}
for all $f$ that satisfy the null hypothesis of $G$-invariance.

\subsubsection{Choices of $t$ and of $\mu_g$}

The test presented here works for any choice of $t$ and variable $g$, though particular choices of these will affect the power of the test. For example, if we choose $t$ such that $p_t = 1$ then our $p$-value will always be $1$ (as all the mass of $Z$ will be on $Z = m$). Similarly if $g = e \in G$ almost surely then we will also only reject with probability at most $\alpha$. \\
\\
For the choice of $t$, we suggest calculating $N_t^g$ from the sample of $D_{I(j)J(j)}^{g_j}$ at some grid of $t$ values $t_0 < t_1 < \cdots < t_k$ with the values of $p_{t_i}$ spread over the interval $(0,1)$, and then taking the $p$-value of $H_0$ as the minimum $p$-value of each $N^g_{t_i}$. This is justified as the information of the test is entirely contained in the set $\{ D_{I(j)J(j)}^{g_j} \}_{j = 1}^m$, i.e., since the $p\text{-value}$ is at most $\PP( Z \geq N_t^g \mid N_t^g )$ for all $t$, we can take an infimum over $t$. \\
\\
For the choice of $\mu_g$, we suggest using a uniform distribution only on some set of topological generators, i.e., a set $A$ with $\overline{ \langle A \rangle } = G$. This means that we don't sample the identity or other elements that generate only subgroups of $G$ to which $f$ may be invariant, but if there is an element that breaks the invariance then one of the generators will too (by lemma \ref{lem:topo_gens_invariance} in appendix \ref{app:proof4}) so we should capture that chance.

\subsubsection{Computational Complexity}
\label{sssec:avt_comp}
The main computational work is finding the nearest neighbour to each resampled point $g_j \cdot X_{I(j)}$. A naive algorithm would search across all $n$ samples to find $J(j)$, which means that the algorithm would cost $O(mn)$ distance calculations. One improvement would be to store the feature vectors $X_i$ in a $k$-$\dd$ tree \citep{bentley1975multidimensional}. Building such a tree takes $O( n \log n)$ operations, but can find nearest neighbours to a new point in $O( \log n)$ time. This speeds up the asymmetric variation to $O( m \log n )$.

\subsubsection{Consistency of this test}
\label{ssec:consistency}

Under the conditions above and some mild conditions on the noise distribution, and when the number of resamples $m$ increases slower than the true sample size $n$, we can prove that the asymmetric variation test is consistent. The condition on the support of $\mu_X$ amounts to restricting $\mathcal{X}$ to the closure of the support as a practitioner would usually do. The condition that the noise admits a density is satisfied in many usual cases in regression (e.g. Gaussian noise). The condition that we can bound the concentration of the noise tightly is somewhat restrictive, but reflects the difficulty of the problem - if the asymmetry is obscured by more noise then it is much more difficult to identify it. The condition on $\mu_g$ ensures that we sample from enough of $G$ to spot points where $f$ is not $G$-invariant. A sufficient condition for this is the sample uniformly from some set of topological generators as suggested above, or from the uniform distribution on $G$ when this exists. 

\begin{prop}[Consistency of the Asymmetric Variation Test]
	\label{prop:cons}
	Set $m$ such that $m \rightarrow \infty$ and $m / n \rightarrow 0$ and fix $t > 0$. Suppose that the law of $X$ has a dense support on $\mathcal{X}$. Suppose that $\epsilon$ admits a density $f_\epsilon$ with respect to Lebesgue measure on $\RR$ that is decreasing in $|\epsilon|$. Suppose that $\mu_g$ is chosen such that $\PP_f( f( g \cdot X) \neq f(X) ) > 0$ for any $f$ that is not $G$-invariant. Then the asymmetric variation test is consistent, i.e. $p_{val} \overset{\PP_f}{\rightarrow} 0$ for all $f$ that are not $G$-invariant.
\end{prop}

\subsection{Permutation Variant of the Asymmetric Variation Test}
\label{ssec:assumptionless_v2}

The asymmetric variation test (in the previous subsection) relies on both the existence of, and the knowledge of, the bounds $V(x,y)$ and $p_t$. In this section we show that we can relax these assumptions and sometimes gain statistical power, but with increased computational cost and without the theoretical guarantee of consistency.  \\
\\
Suppose $\mu_g$ is a distribution on $G$ such that $d_{\mathcal{X}}( g_1 \cdot X_1, g_2 \cdot X_2 ) \overset{D}{=} d_{\mathcal{X}}(X_1, X_2)$ when $g_i \iid \mu_g$ and $X_i \iid \mu_X$. If $\mu_X$ is $G$-invariant, then any distribution on $G$ suffices. Such a $\mu_g$ always exists by placing all mass on the identity but often more diverse possibilities usually exist when the practitioner is free to choose the covariate sampling distribution $\mu_X$. Then under the null hypothesis \textit{$H_0 : f$ is $G$-invariant} we have the equality:
\begin{equation}
	\label{eq:perm_interchangability}
	 \frac{ | Y_i - Y_j | }{ d_{\mathcal{X}}(X_i, X_j) } \overset{D}{=} \frac{ | f( g_i \cdot X_i ) + \epsilon_i - f( g_j \cdot X_j ) - \epsilon_j  | }{ d_{\mathcal{X}}( g_i \cdot X_i, g_j \cdot X_j) } =  \frac{ | Y_i - Y_j | }{ d_{\mathcal{X}}( g_i \cdot X_i, g_j \cdot X_j) } 
\end{equation}
for any variation bound $V$ for $f$. Thus instead of comparing outputs at transformed \textit{points} $g_{I(i)} \cdot X_{I(i)}$ to non transformed points $X_{J(i)}$ as with the asymmetric variation test we now compare the slopes above for transformed \textit{datasets}: 
\begin{equation}
	\mathcal{D}^{(k)} = \{ (g_i^{(k)} \cdot X_i, Y_i ) \}_{i = 1}^n
\end{equation}
where $g_i^{(k)} \iid \mu_g$ for all $i$ and $k \in \{ 1, \dots, B \}$. For any non $G$-invariant function $f$ we break the second equality of \ref{eq:perm_interchangability}, and when there is a bound $V$ as with the previous test then the right hand side is frequently larger than the left. This leads to a resampling test with minimal assumptions defined in algorithm \ref{algo:perm_2}.  \\

\begin{algorithm}[h]
\caption{Permutation Variant of Asymmetric Variation Test}
\label{algo:perm_2}
\begin{algorithmic}[1]
\Procedure{PermVarTest}{$\mathcal{D}$, $\mu_g$, $B$}
	\For{$k \in \{1, \dots, B \}$}
		\State Sample $g_i \iid \mu_g$ for $i \in [n]$
		\State Compute $S^k = \{  | Y_i - Y_j | / d_{\mathcal{X}}( g_i^{(k)} \cdot X_i, g_j^{(k)} \cdot X_j)  : i \neq j \in [n] \}$  
		\State $A(k) \leftarrow \max S^k$ 
	\EndFor
	\State Compute $S^0 = \{ | Y_i - Y_j | / d_{\mathcal{X}} ( X_i, X_j ) : i  \neq j \in [n] \}$
	\State $A_0 \leftarrow \max S_0$
	\State $p_{val} \leftarrow |\{ k : A(k) \leq A_0 \}| / B$
	\State \Return $p_{val}$
\EndProcedure
\end{algorithmic}	
\end{algorithm}

Under the null hypothesis equation \ref{eq:perm_interchangability} ensures interchangeability of the variables $A(k)$ and $A_0$, so the type $1$ error of the hypothesis test $\mathrm{Test}_{\alpha}( \mathcal{D}, G ) = 1 - 2 \mathbf{1}_{ p_{val} \leq \alpha }$ is controlled at level $\alpha$ in finite samples. We do not prove consistency of this test but note in the numerical experiments (section \ref{sec:exp}) that the power is often higher than the previous test, particularly for low sample sizes and in high dimensions. 
%
%
%
%

\subsubsection{Choice of sampling distribution on the group}
\label{sssec:choice_of_mu_g}

The practitioner has a choice of the sampling method from the group $G$ as long as it satisfies equation \ref{eq:perm_interchangability}. As stated above, if $\mu_X$ is $G$-invariant then we can use $\mu_g = U(G)$, and so we suggest that when possible these choices are made. For example, if $\mathcal{X}$ is a compact domain such as the hypersphere $S^d = \{ x \in \RR^{d + 1} : \| x \|_2 = 1 \}$ then we would typically try to sample the covariates $X_i$ from $U(S^d)$ to sufficiently cover the space. This distribution is invariant to every measure on $SO(d)$ so we can pick either the Haar measure $U( SO(d) )$. We could also use a measure supported on a finite set of topological generators, as suggested in the previous test, or any other desired measure in this case. \\
\\
In other cases it may be hard to satisfy this condition whilst also sampling from topological generators of $G$. Indeed for both these conditions to be true we require that $\mu_X$ is $G$-invariant which does not always hold. Instead we note that the distributional equality $g \cdot X \overset{D}{=} X$ is not strictly \textit{necessary} for equation  \ref{eq:perm_interchangability} and the exactness of the test. Instead, the suitability of a particular distribution can be examined by comparing the marginal distribution of $d_\mathcal{X} ( g_i \cdot X_i, g_j \cdot X_j )$ against $d_{\mathcal{X}}( X_i, X_j)$. We do this by visual inspection of $Q$-$Q$ plots and Kolmogorov-Smirnov tests for the equality of these distributions.

\subsubsection{Computational Complexity}
\label{sssec:pv_comp}
This permutation method requires $O(nB)$ samples and transformations of the inputs $X_i$ and then the computation  of $O(n^2B)$ distance matrices. The remaining computations (the computing the maxima and the comparisons) are negligible in comparison. Thus the overall complexity is $O(n^2B)$ which is at least $B$ times more intensive than the Asymmetric variation test (with the choices suggested in its section).

\section{Numerical Experiments}
\label{sec:exp}

We now demonstrate the performance of the confidence region $\hat{K}_\ell$ on synthetic datasets. The first simulation is a low dimensional regression estimation where we can directly quantify the error $\mathcal{E}( \hat{K}_\ell )$ directly. The second simulation then considers the effect of using the estimated group $\hat{G}_\ell$ to symmetrise an estimator $\hat{f}$ of $f$ and consider the mean squared prediction error. A summary of the proportions of correct group estimations (i.e., where $\hat{G}_\ell = G_{\max}( f, K_\ell)$) when $\hat{G}_\ell$ is sampled uniformly from $\max \hat{K}_\ell$ is shown in table \ref{tab:proportions_of_correct_hat_G}. \\ 
\\

\begin{table}
\caption{\label{tab:proportions_of_correct_hat_G} Proportions of correct symmetry estimation.}
\fbox{\begin{tabular}{l l l | r r r r r}
	\textbf{Example}		&	\textbf{Scenario}	& \textbf{Test}    & $n = $ \textbf{20} & \textbf{50} & \textbf{100} & \textbf{200} & \textbf{500} \\ \hline 
\multirow{ 6}{*}{\ref{eg:low_dim_G_sims}} & \multirow{2}{*}{Dimension 2} & AVT & 0.250 & 0.745 & 0.960 & 1 & 1 \\ 
						&  & PV & 0.480 & 0.700 & 0.750 & 0.735 & 0.700 \\ 
						& \multirow{2}{*}{Dimension 4} & AVT & 0.010 & 0.03 & 0.045 & 0.265 & 0.870 \\ 
						&  & PV & 0.135  & 0.50  &  0.795  & 0.855 &  0.905  \\ 
						& \multirow{2}{*}{Dimension 10} & AVT & 0  & 0 & 0 & 0 & 0 \\
						&  & PV & 0.03 & 0.06 & 0.125 & 0.27 & 0.565  \\
						\hline 
\multirow{ 4}{*}{\ref{eg:est_f_sims}} 		& Spherically Symmetric & PV & 	0.85 & 0.90 & 0.83 & 0.83 & - 		\\
						& Rotationally Symmetric & PV & 0.51 & 0.97 & 0.93 & 0.97 & - 			\\
						& Totally Asymmetric & PV & 0.10 & 0.85 & 1 & 1 & -  			\\
						& Misspecified Lattice & PV & 0.09 & 0.34 & 0.91 & 1 & -  
\end{tabular} }
\end{table}


All tests were conducted locally on a 2020 13'' MacBook Pro laptop with a 2 GHz Quad-Core Intel Core i5 processor and 16GB of RAM. All code is available on GitHub (\href{https://github.com/lchristie/estimating_symmetries}{https://github.com/lchristie/estimating\_symmetries}). 

\begin{eg}[Finite Group in Low Dimensions]
	\label{eg:low_dim_G_sims}
	Consider the regression functions:
	\begin{equation}
		f_d : (x_1, \dots, x_d) \mapsto \exp( - | x_1 | )
	\end{equation}
	for each dimension $d = 2,3,4, \dots$ and suppose that $X_i \iid N(0,4I_d)$ and $\epsilon_i \iid N(0, 0.05^2 )$. Consider the group $D_4$ acting by rotations and reflections of the $x_1$-$x_2$ plane. We consider the lattice base:
	\begin{equation}
		\big\{ G_i \big\}_{i = 1}^\ell = \big\{ \langle R_h \rangle , \langle R_v \rangle, \langle R_\pi \rangle, \langle R_/ \rangle, \langle R_\backslash \rangle , \langle R_{\pi / 2} \rangle \big\}
	\end{equation}
	which has $K_\ell = L(D_4)$. Note that $R_h$ reflects through the $x_2$ line in this $x_1$-$x_2$ plane, and $R_v$ through the $x_1$ line etc. Then the symmetries $R_h$, $R_v$, and $R_\pi$ all apply to $f_d$ as they only change signs of $x_1$ and $x_2$. However $ R_/$,  $R_\backslash$, and $R_{\pi / 2}$ do not because they exchange $x_1$ and $x_2$ (as well as change signs). Therefore the correct object we wish to estimate is:
	\begin{equation}
		G_{\max}( f_d, K_\ell ) = \langle R_h, R_\pi \rangle
	\end{equation}
	for every dimension $d$, and the correct test result vector would be $(1,1,1,0,0,0)$. \\
	\\
	We tested each null hypothesis \textit{$H_0^{(i)} : f$ is $G_i$-invariant} at the $\alpha_i = 0.05$ significance level. This means that $\hat{K}_\ell$ will be a $0.7$ confidence region for $G_{\max}(f, K_\ell)$ per theorem \ref{thm:cons_of_hat_G}, but we see that it contains this object much more frequently in appendix \ref{app:correct_containment}. Each sampling distribution from the groups are uniform on the group. The asymmetric variation tests had $t = 2 \sigma = 0.1$, $p_t = \frac{ 2 \sigma }{ t  } \tfrac{ \exp( - t^2 / 4 \sigma^2 )}{\sqrt{2\pi}}$, and the number of resamples equal to the number of samples. The permutation variants used $B = 100$. \\
	\\
	We plotted the results of these simulations for dimensions 2 and 4 in figure \ref{fig:toy_eg_K_hat} for the performance of the confidence region $\hat{K}_\ell$. Note that if $\mathcal{E}( \hat{K}_\ell ) = 0$ then $\hat{G}_{\ell} = G_{\max}( f, K_\ell )$. These plots show that the Asymmetric Variation Test can detect this asymmetry reliably, though this performance decays with higher dimension as the $X_i$ become spaced with further distances. The permutation variant interestingly struggles more in low dimensions, exhibiting moderate errors even when $n = 300$, though outperforms the AVT in dimension 4 and above. Results for the higher dimensional simulations are available in appendix \ref{app:higher_dimension_eg_61}.   \\
	\\
	 In this example computations of the asymmetric variation test when $n = 100$ and $d = 4$ each took 0.004s and computations of the permutation variant each took 0.118s.
	
	\begin{figure}[h]
		\centering
		\begin{tabular}{ m{7cm} m{7cm} }
			\centering \includegraphics[scale=0.48]{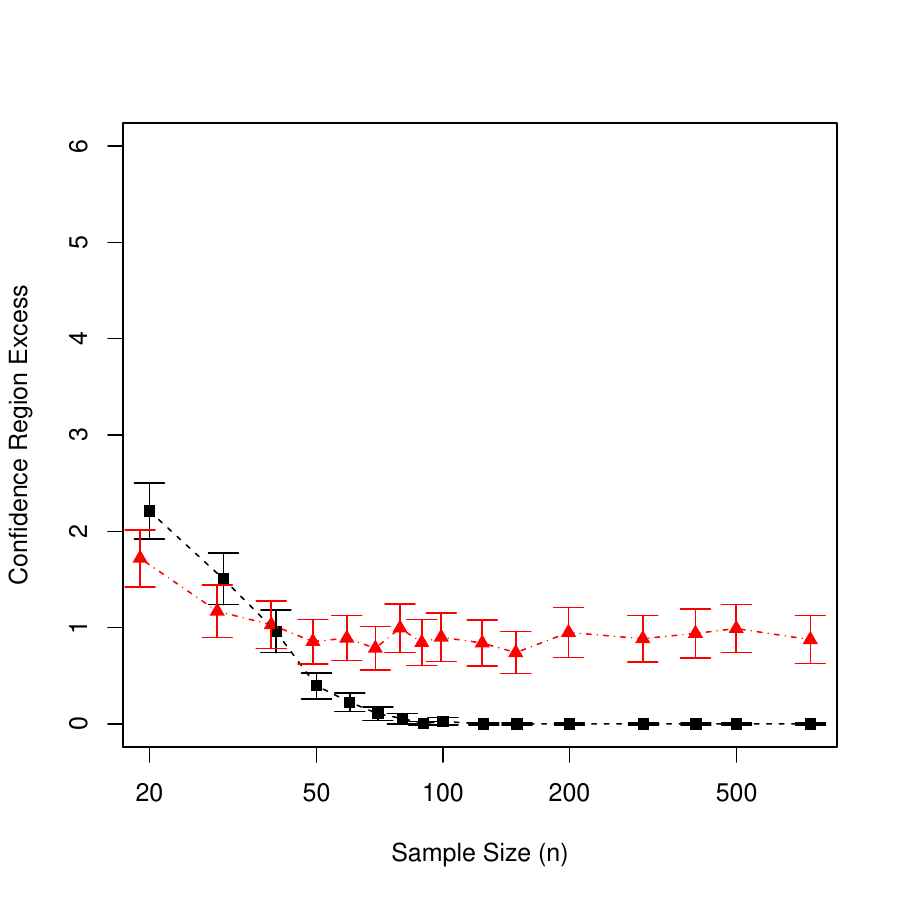} &
			\centering \includegraphics[scale=0.48]{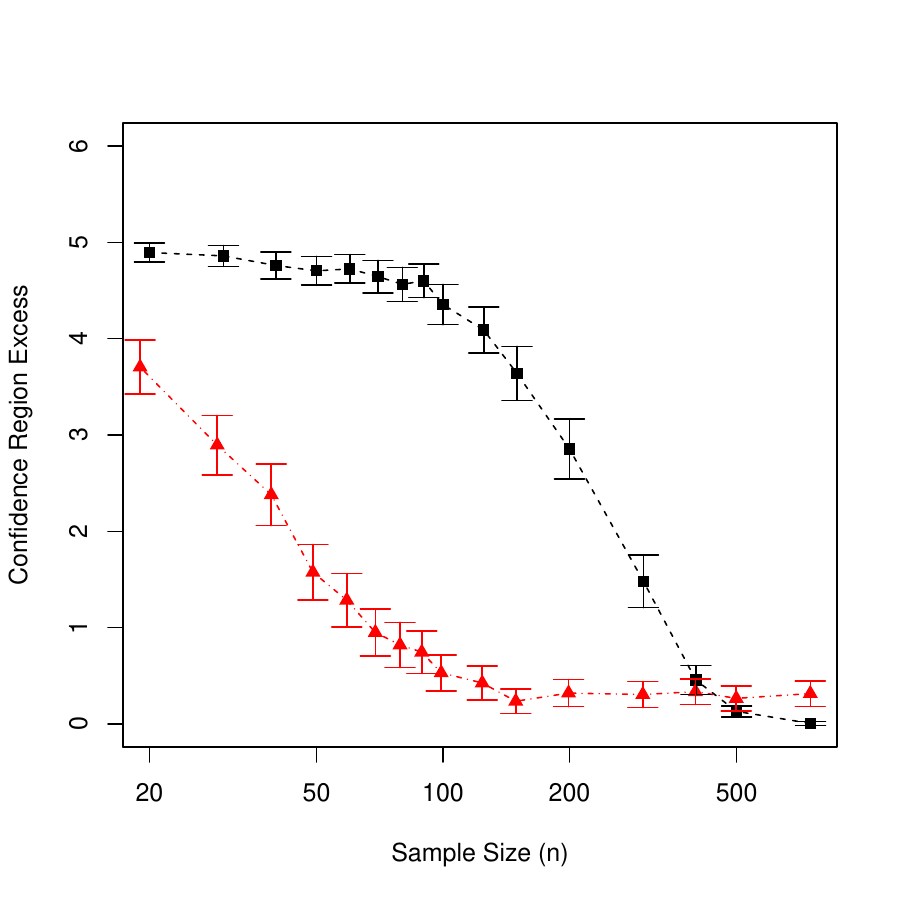} \tabularnewline
			\centering (a) $\mathcal{E}( \hat{K}_\ell )$ when $d = 2$ & 
			\centering (b) $\mathcal{E}( \hat{K}_\ell )$ when $d = 4$
		\end{tabular}
		\caption{This figure describes the results of the simulations in example \ref{eg:low_dim_G_sims}. Subfigures (a) and (b) track the confidence region excess $\mathcal{E}( \hat{K}_\ell )$ for each test used (black for the asymmetric variation test and red for the permutation variant), in dimension 2 for (a) and dimension 4 for (b). The error bars are Wald 95\% Confidence Intervals for the true expected excess. }
		\label{fig:toy_eg_K_hat}
	\end{figure}
\end{eg}

\begin{eg}[Simulation of Symmetrised Estimators]
\label{eg:est_f_sims}
We now consider an example of estimating continuous symmetries. Suppose that the domain $\mathcal{X}$ is $\RR^3$ and let $\mathcal{G} = SO(3, \RR)$ act on $\mathcal{X}$ by matrix multiplication. Let $K_\ell$ be the sub-lattice of $K(SO(3, \RR))$ with lattice base:
\begin{equation}
	\big\{ G_i \big\}_{i = 1}^\ell = \big\{ S^1_{e_1}, S^1_{e_2}, S^1_{e_3}, S^1_{u_1}, \dots, S^1_{u_6}, SO(3) \big\}.
\end{equation}
Here the axes $u_1, \dots u_6$ (given explicitly in appendix \ref{app:icosahedral_vectors}) correspond to non-parallel unit vectors through the vertices of a regular icosahedron, rotated such that the distances to the standard basis vectors are roughly maximised (as with section \ref{ssec:choose_lattice}). We consider three possible regression functions with differing levels of symmetry:
\begin{center}
	\begin{tabular}{ p{6.1cm} p{3.8cm} p{4cm} }
	\textbf{Regression Function } & \textbf{Maximal Symmetry} & $\phantom{0}$ \\ \hline \hline
	$f_1(x) = \sin( \| x \|_2 )$ & $G_{\max}( f_1, K_\ell ) = SO(3)$ & Spherically Symmetric \\ 
	$f_2(x) = \sin( | x_1 | ) + \cos( \sqrt{ x_2^2 + x_3^2 } ) $ & $G_{\max}( f_2, K_\ell ) = S^1_{e_1}$ & Rotationally Symmetric \\
	$f_3(x) = \sin( | x_1 | ) + 0.2 x_2  - \cos(  x_3 )$ & $G_{\max}( f_3, K_\ell ) = I$ & Totally Asymmetric \\
	$f_4(x) = f_2( R_{0.1, e_3} x  )$ & $G_{\max}( f_3, K_\ell ) = I$ & Symmetry not in $K_\ell$
    \end{tabular} 
\end{center}
where $R_{0.1, e_3}$ is the rotation matrix of 0.1 radians around the $z$ axis. This means that $f_4$ has symmetries that are isomorphic to those of $f_2$, but are not included in the lattice $K_\ell$. This demonstrates the usefulness of these methods even in the case when the lattice $K_\ell$ is mis-specified. \\
\\
As per the statistical problem (Section \ref{ssec:stats}), we simulate data $\mathcal{D} = \{ (X_i, f(X_i) + \epsilon_i \}_{i = 1}^n$ and $\mathcal{D}' = \{ (X_i', f(X_i') + \epsilon_i' \}_{i = 1}^n$ where $X_i \iid N( 0, 2 I_3 )$, $X_i' \iid N( 0, 2 I_3 )$, and $\epsilon_i, \epsilon_i' \iid N(0, 0.01^2)$. \\
\\
We estimate $\hat{G}_\ell$ using the permutation variant of the asymmetric variation test with $B = 100$ and $\mu_g$ sampled uniformly from each group. We then consider the problem of estimation of $f$, using an estimate that has been symmetrised by $\hat{G}_\ell$. In particular, we consider three regressors:
\begin{enumerate}[	(A)]
	\item A standard local constant estimator (LCE) $\hat{f}(x, \mathcal{D})$ with bandwidths selected by cross-validation (the default in the R-package \textsf{NP} \citep{hayfield2008np});
	\item A symmetrised LCE $\hat{f}_{\hat{G}}^A \big( \pi_{\hat{G}}(x) ,  \{ ( \pi_{\hat{G}}( X_i), Y_i ) \}_{i = 1}^n \big)$ (utilising data projection); and
	\item A symmetrised LCE that splits $\mathcal{D}$ into independent pieces half the size of $\mathcal{D}$, then using the first half to estimate $\hat{G}_\ell$ only and the second half to estimate $\hat{f}$ afterwards as with (B).  
\end{enumerate}

Here $\pi_G : \mathcal{X} \rightarrow  \mathcal{X} / G $ is the natural projection given by $\pi_{I}(x) = x$, $\pi_{SO(3)}(x) = \|x \|_2$, and $\pi_{SL(3)}(x) = \mathbf{1}_{x \neq 0}$. The natural projection for the circle groups are given by $\pi_{S^1_u} (x ) = ( x \cdot u , \sqrt{ \| x \|^2 - ( x \cdot u )^2 } )$. We have chosen local constant, rather than local linear for example, estimators because of computational time and because we do not have boundary effects with the domain. \\
\\
The mean squared prediction errors are plotted in figures \ref{fig:mspes_1} and \ref{fig:mspes_2}. We see that for scenarios 1 and 2 the MSPE decays rapidly and is able to correctly generalise once $\hat{G}$ is more likely than not to estimate the correct symmetry. When there is no symmetry present, as with scenario 3, the we quickly see convergence of $ \hat{f}_{\hat{G}}^A - \hat{f} \overset{\PP_f}{\rightarrow} 0$. Interestingly, we see no empirical benefit from splitting the data into independent pieces, suggesting that the errors of the group estimation and the functional estimation are uncorrelated. In the fourth scenario at low sample sizes, using the errant symmetry to symmetrise the estimator reduces the variance more than it increases bias resulting in smaller MSPE. At higher sample sizes it correctly rejects the $S_{e_1}^1$ symmetry and the estimator converges to the baseline. 

\begin{figure}[h]
	\centering
	\begin{tabular}{ m{7cm} m{7cm} }
		\centering \includegraphics[scale = 0.48]{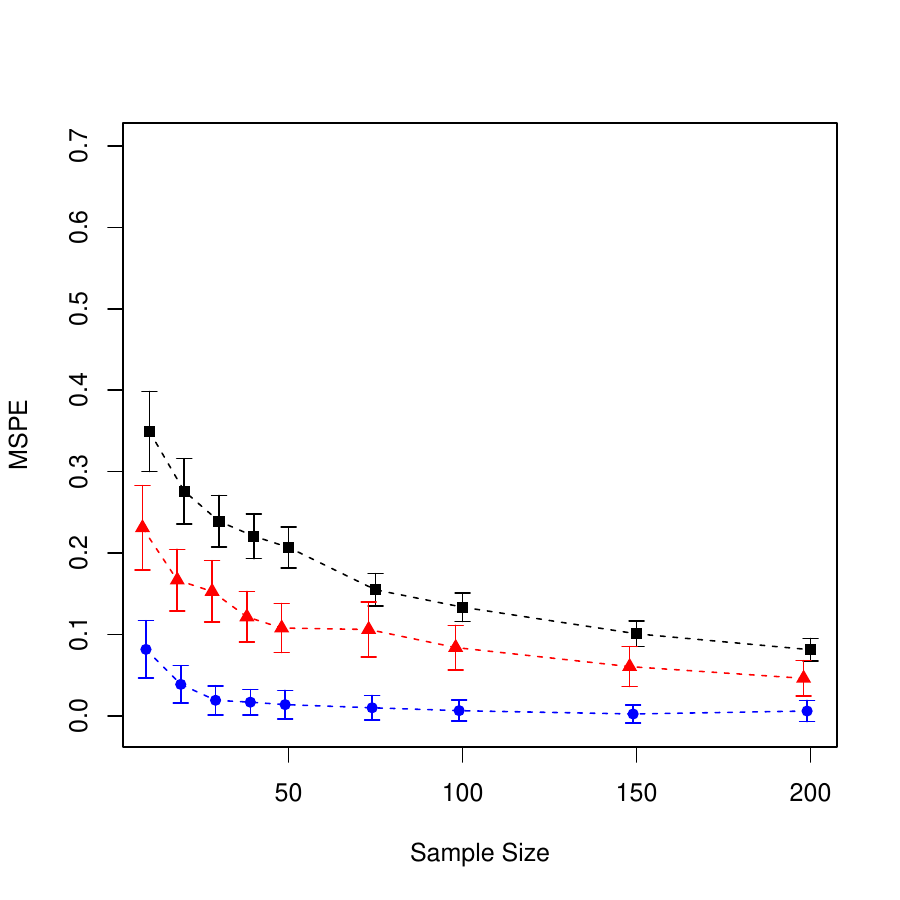} &
		\centering \includegraphics[scale = 0.48]{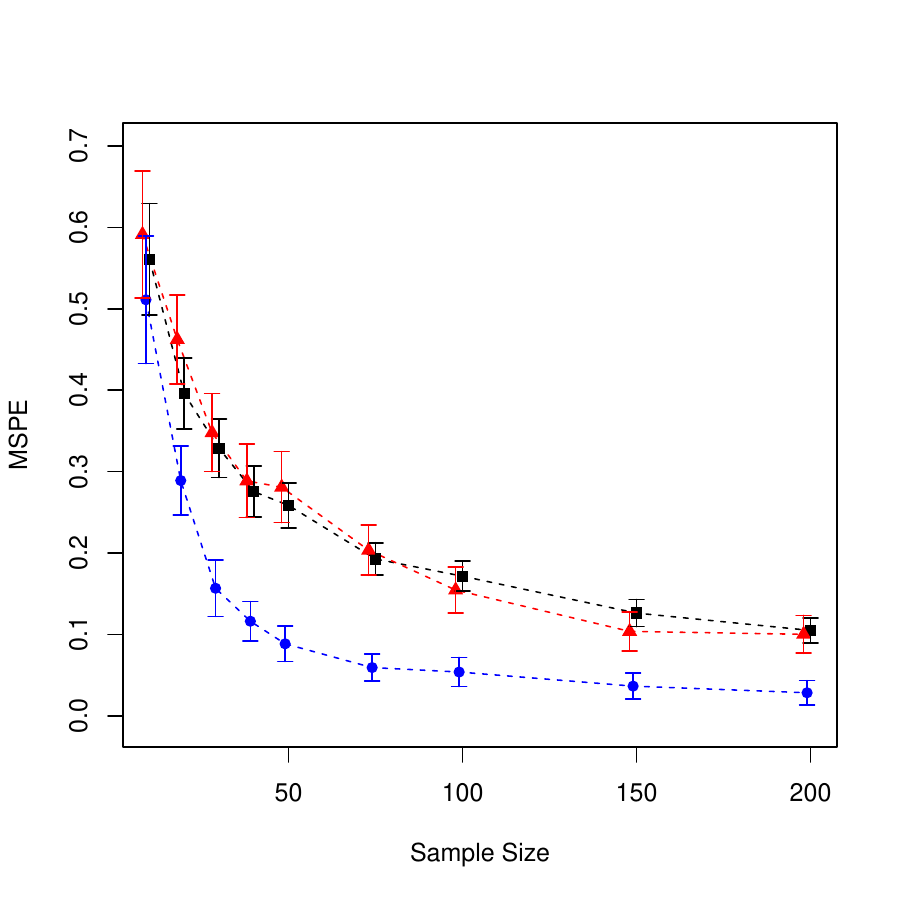} \tabularnewline 
		\centering (a) Scenario 1 & 
		\centering (b) Scenario 2 
	\end{tabular}

%
	\caption{Mean squared prediction errors for the three estimators, under scenarios 1 and 2. The averages over 100 simulations are shown in black for $\hat{f}$, in blue for $\hat{f}^A_{\hat{G}}$, and in red for $\hat{f}^S_{\hat{G} }$. The error bars are Wald 95\% Confidence Intervals for the true mean squared predictive error.}
	\label{fig:mspes_1}
\end{figure}

\begin{figure}[h]
	\centering
	\begin{tabular}{ m{7cm} m{7cm} }
		\centering \includegraphics[scale = 0.48]{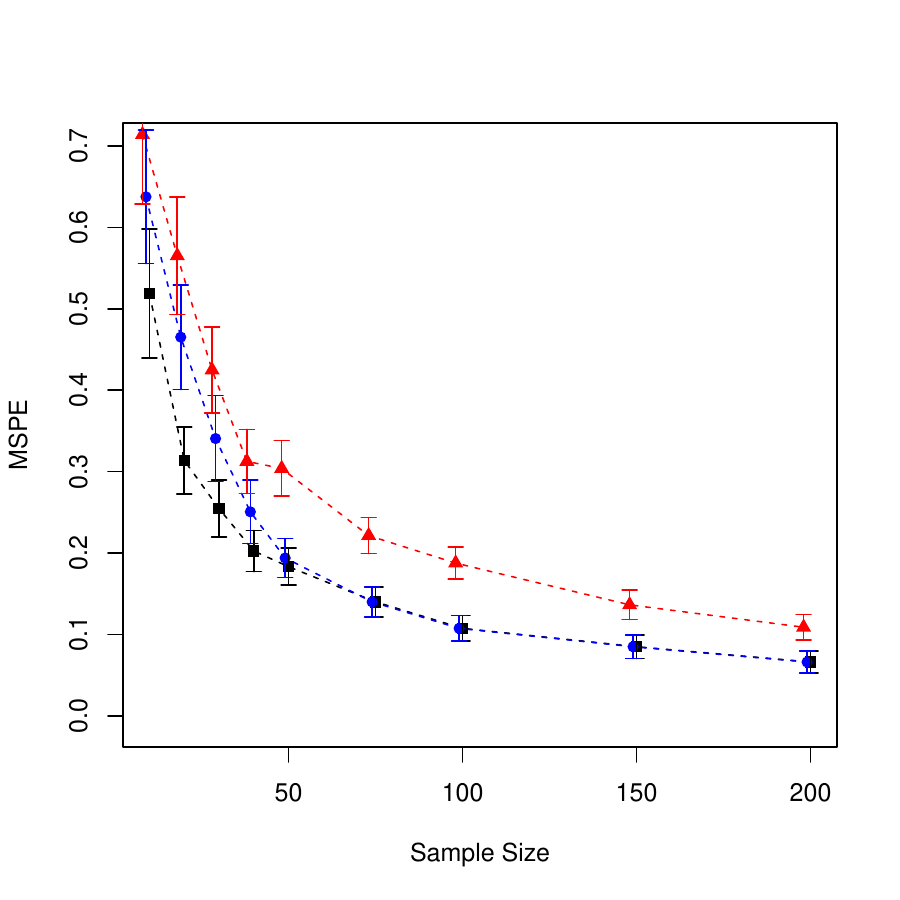} &
		\centering \includegraphics[scale = 0.48]{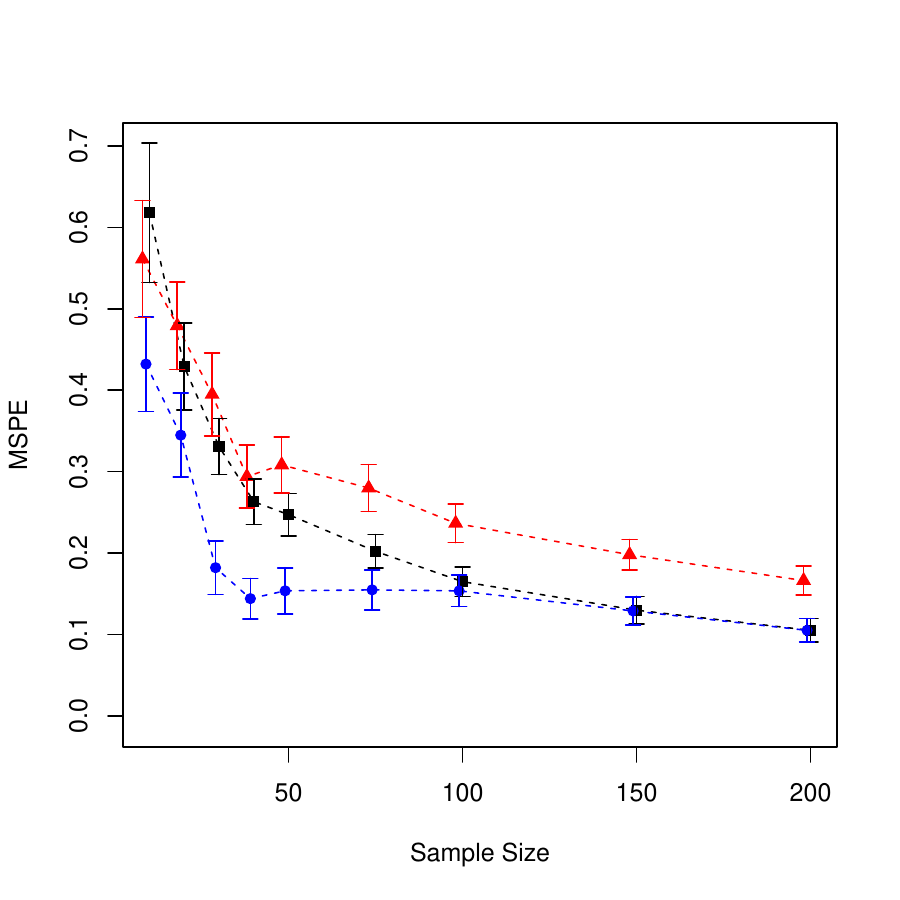} \tabularnewline 
		\centering (a) Scenario 3 & 
		\centering (b) Scenario 4 
	\end{tabular}

%
	\caption{Mean squared prediction errors for the three estimators, under scenarios 3 and 4. The averages over 100 simulations are shown in black for $\hat{f}$, in blue for $\hat{f}^A_{\hat{G}}$, and in red for $\hat{f}^S_{\hat{G} }$. The error bars are Wald 95\% Confidence Intervals for the true mean squared predicted error.}
	\label{fig:mspes_2}
\end{figure}

\end{eg}

\section{Applications to Real Data}
\label{sec:app}

We now demonstrate these methods on two real world datasets. The first is from geophysics, where we seek to estimate the Earth's magnetic field. In this example, we show that our confidence region built from satellite readings contains the identity only. This supports the practice of modelling an asymmetric magnetic field, as in \cite{thebault2015international}. The second example explores the spatio-temoral relationship of sunspots. We show that in this example our estimated maximal symmetry of this relationship is rotational, and consistent with the phenomenon known as Sp\"{o}rer's law \citep{babcock1961topology} for most of the 12 available solar cycles. The other four cycles have $\hat{G}_\ell = I$, which indicates that there are longitudinal effects along with the latitudinal effects in Sp\"{o}rer's law.

\subsection{Satellite Based Magnetic Field Data}
\label{ssec:magnets}
The European Space Agency (ESA) launched the SWARM mission in 2013 to measure the earth's magnetic field and other related properties. This consists of three satellites orbiting between 450km and 530km above the earth's surface. These measurements complement the observations at ground based observatories to better model the field for use in navigation and geospatial modelling. This vector field $B: \RR^3 \rightarrow \RR^3$ is generated primarily by the rotation of the inner core, but is also affected by: magnetised material in the earth's crust; electrical currents in the upper atmosphere caused by solar winds; and the interaction of the sun's magnetic field. Without these additional effects it is reasonable to model the field as one would a bar magnet, with a rotational symmetry through the magnetic poles.   \\
\\
Given these effects can break this symmetry, the vector field $B$ on and above the earth's surface at time $t$ is typically modelled as the gradient of a potential field $B = \nabla V$ where: 
\begin{equation}
	V(r, \theta, \phi, t) = a \sum_{n = 1}^{13} \sum_{m = 0}^n \left( \frac{a}{r} \right)^{n+1} \Big( g_n^m(t) \cos( m \phi ) + h_n^m(t) \sin( m \phi) \Big) P_n^m( \cos(\theta) )
\end{equation}
in the International Geomagnetic Reference Field (IGRF) \citep{thebault2015international}. Here $a = 6,371.2$km is the earth's radius, $\theta$ is the geocentric co-latitude, $\phi$ is the east longitude. The functions $P_n^m$ as the Schmidt quasi-normalised associated Legendre functions of degree $n$ and order $m$. The Gauss coefficients $g_n^m$ and $h_n^m$ are then fitted using observational data from ground based stations. \\
\\
Using the data of the ESA SWARM satellites, accessed via the VirES client (\url{https://vires.services/}),  we can estimate the symmetries of the field. A full description of the data source is available in appendix \ref{app:vires}. We use 1000 measurements of the field intensity $\| B \|_2$ as the response variable chosen uniformly at random from the observations on the 25th of February 2023. These measurements are plotted in figure \ref{fig:magnetic_field}. \\

\begin{figure}[h]
	\centering
	\begin{tabular}{m{6cm} m{8cm} }
		\centering	\includegraphics[scale = 0.4]{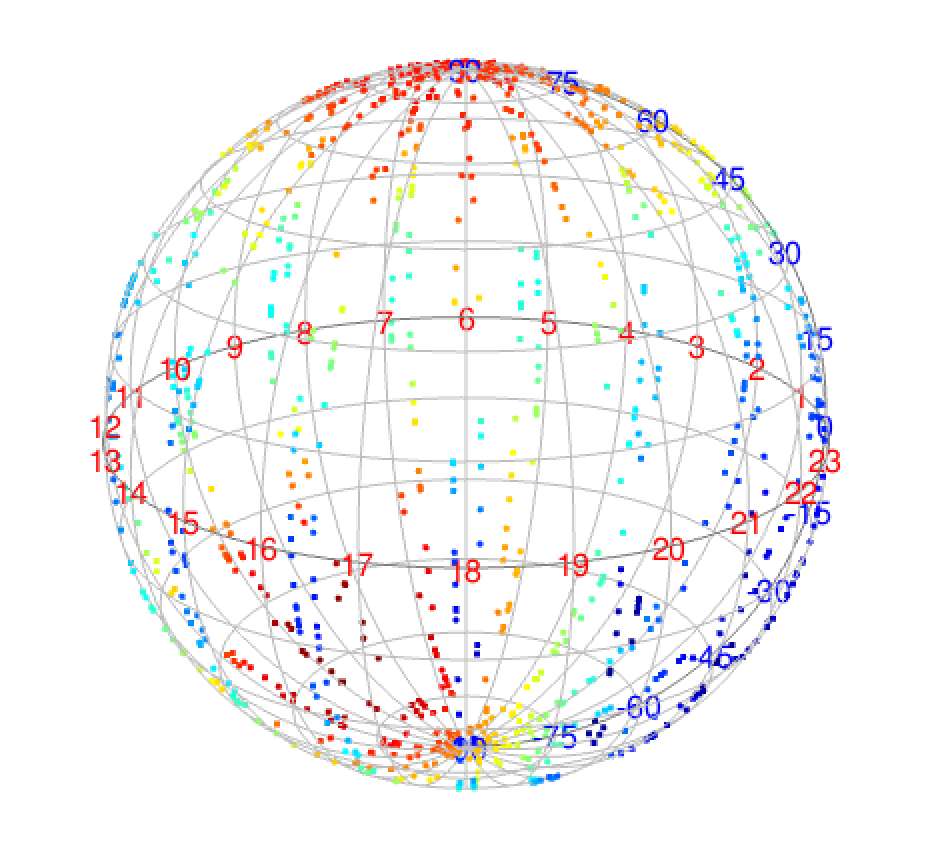} &
		\centering \includegraphics[scale = 0.5]{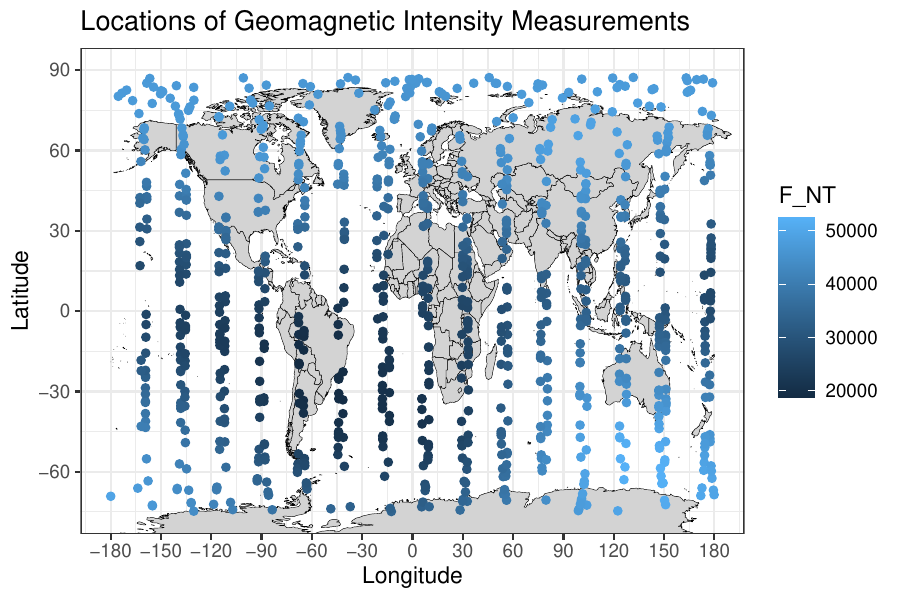} \tabularnewline
        \centering (a) Geo-spatial plot & 
        \centering (b) Mercator Projection
	\end{tabular}
	\caption{A visualisation of the measurements of the earth's magnetic field, using the 1000 samples on February 25, 2023. In subfigure (a), Red dots represent the high intensity measurements that occur towards the poles, and blue dots the regions of lowest intensity towards the equator. Red numbers indicate longitude and blue numbers represent latitude. In subfigure (b) intensity high magnetic field intensity is displayed in light blue and low intensity in dark blue.}
	\label{fig:magnetic_field}
\end{figure}

We have used the same symmetry estimation procedure of example \ref{eg:est_f_sims}, with six additional symmetries in the lattice base (corresponding to vertices of a further offset icosahedron). These axes of these $S^1_u$ rotational symmetries are given explicitly in appendix \ref{app:icosahedral_vectors}. Unlike example \ref{eg:est_f_sims}, the data is not distributed symmetrically around the sphere, and examinations of the $Q$-$Q$ plots in appendix \ref{app:qq_plots_magnets} and their associated $p$-values show that uniform distributions on the group are unsuitable. Instead we sample from each group by sampling (small) angles from $N(0, 0.05^2)$ in radians and constructing rotations around the group's axis of these angles. The $Q$-$Q$ plots and KS tests show that these distributions are suitable. \\
\\
All symmetries were rejected at the $\alpha = 0.05$ significance level, so the estimated maximal symmetry is just trivial group $I$. This supports the modelling approach of the IGRF and rejects the simpler bar magnet model of the earth's magnetic field.

\subsection{Sunspots and Sp\"{o}rer's Law}
\label{ssec:sunspots}

The sun's photosphere has small regions of low intensity, known as \textit{sunspots}. These sunspots occur due to complex interactions between aspects of the inner and outer solar magnetic fields, as described in \citet{babcock1961topology}. As noted in \citet{garcia2020optimal}, the distribution of the locations of these spots over each orbit is approximately symmetric around the rotational axis of the sun. However, this distribution depends on time within each solar cycle (the period between reversals of the polarity of the magnetic field) - the latitudes decrease over the course of the period. This phenomenon, known as \textit{Sp\"{o}rer's Law}, is explained by an intensification of the twisting that causes the sunspots that culminates in a polarity reversal.  \\
\\
Records of the sunspots are readily available in the R package \textsf{rotasym} \citep{garcia2020optimal}. This was originally collected in Debrecen Photoheliographic Data (DPD) sunspot catalogue \citep{baranyi2016line}. We have taken a uniform sample of 250 sun spots from the 21st solar cycle and taken a local average of time of occurrence of sunspots from the entire catalogue of that cycle. This data is represented as the vector $(x_i, y_i, z_i, t_i)$ with $x_i^2 + y_i^2 + z_i^2  = 1$. The locations of the spots in this orbit are shown in figure \ref{fig:sun_spot_plot}.   \\

\begin{figure}[h!]
	\centering
	\begin{tabular}{ m{7cm} m{6cm} }
		\centering \includegraphics[scale = 0.25]{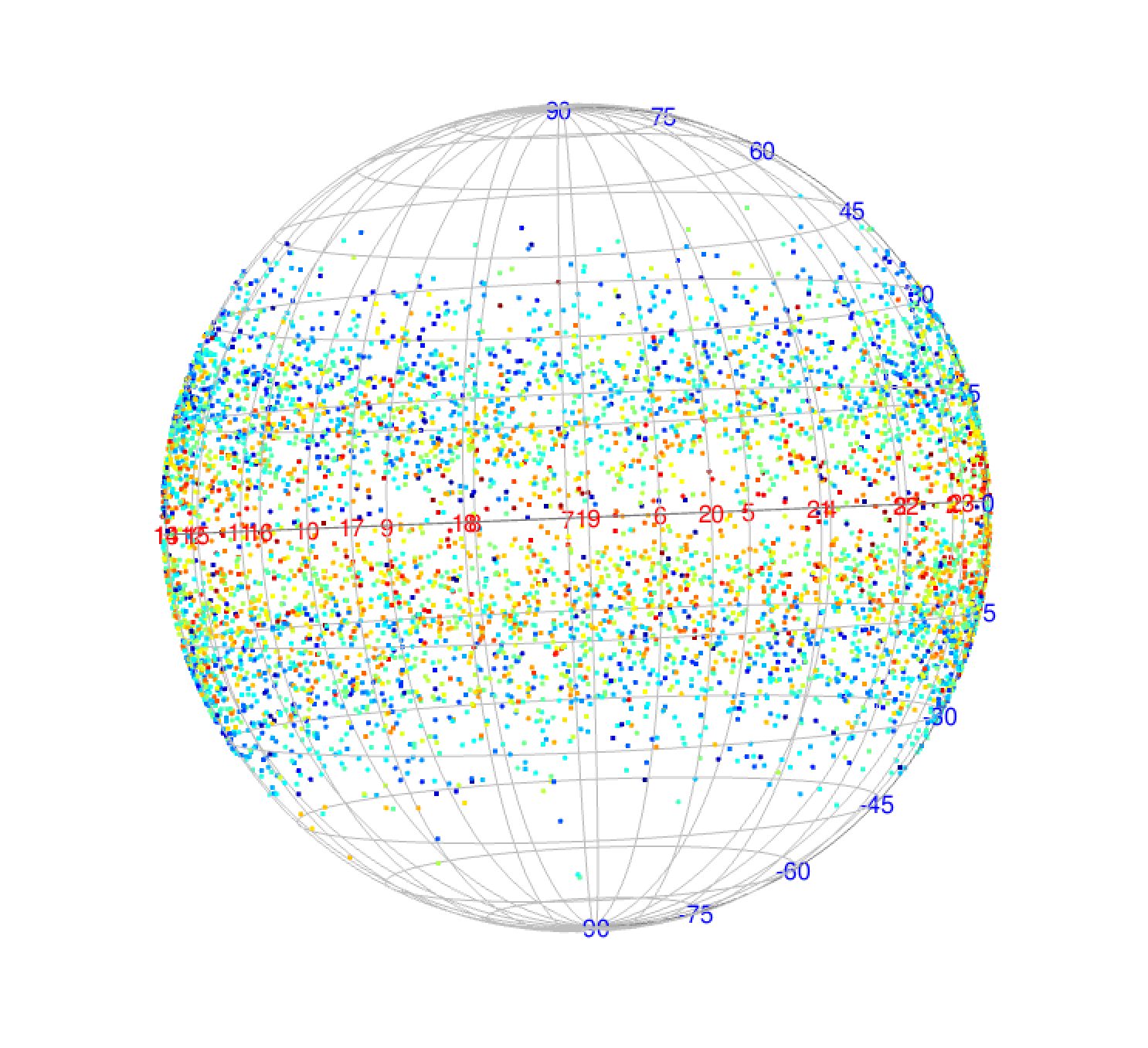} & 
		\centering \includegraphics[scale = 0.4]{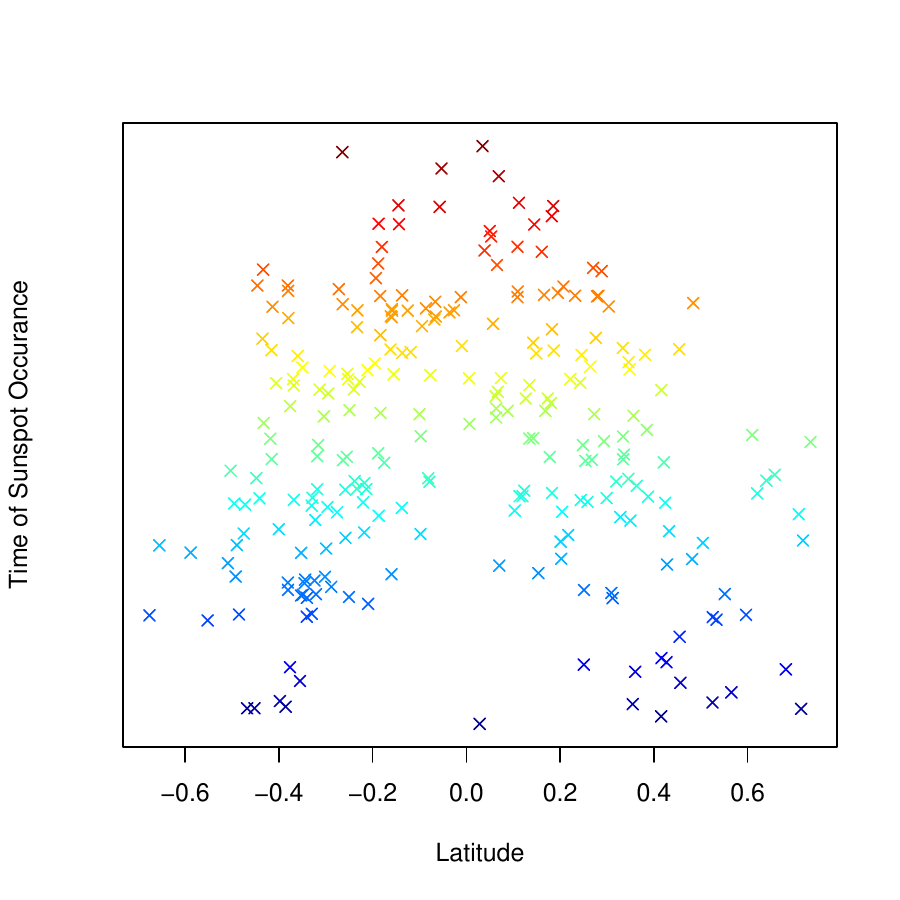} \tabularnewline
		\centering	(a) Helio-spatial Plot & 
		\centering (b) Plot of time vs. latitude
	\end{tabular}
	\caption{Locations of all sunspots during the 21st recorded solar cycle, from the \textsf{rotasym} R package. The colours represent the time with the cycle, with blue towards the start and red towards the end. The blue numbers in subfigure (a) indicate latitude, the red longitude in degrees. Also shown is a marginal of time against latitude in subfigure (b) for 250 points uniformly sampled from this cycle.}
	\label{fig:sun_spot_plot}
\end{figure}
We can examine the symmetries of this relationship between the locations $X = (x, y, z)$ of the observed spots and the mean time of occurrence conditioned on locations $f(x,y,z) = \EE( t \mid x,y,z)$. Whilst this is not the usual regression situation, where we would predict location given the time, it is mathematically equivalent. Specifically, if the spatial distribution conditioned on time is invariant to $G$, then the conditional mean $\EE( t \mid x,y,z)$ must be too\footnote{This is seen by considering first the marginal of the spatial distribution (which must be $G$-invariant), and then applying Bayes Law and the invariant of the conditional and the marginal as required.}, and so we can use the symmetry testing and estimation techniques to give the largest possible symmetry of this conditional distribution. \\
\\
We have conducted tests on the same lattice base as in the previous application except for $SO(3)$, so the lattice base is:
\begin{equation}
	\{ G_i \}_{i = 1}^{15} = \{ S_{u_1}, S_{u_2}, \dots, S_{x}, S_{y}, S_{z} \}
\end{equation} 
The vectors $u_i$ are given in appendix \ref{app:icosahedral_vectors}. We used the permutation variant with $B = 100$. Sampling from each group was done by sampling an angle from $N( 0, 0.05^2 )$ and constructing the rotation around the axis by this angle. As suggested in section \ref{sssec:choice_of_mu_g}, we checked $Q$-$Q$ plots and tested for differences between the original and transformed pairwise distance distributions using a Kolmogorov-Smirnov test, which showed that these sampling distributions are suitable. These are available for this cycle in appendix \ref{app:qq_plots_sunspots}. A uniform distribution on the angles was un-suitable because of the deviation from the original pairwise distance distribution for all groups except for $S_z^1$.   \\


For the $21^{st}$ cycle all tests on the lattice base were rejected at the $\alpha_i = 0.05$ significance level except for the rotational symmetry around the $z$-axis, corresponding to a longitudinal symmetry. Thus the estimated maximal symmetry is $S_z^1$. This means that the data is consistent both with Sp\"{o}rer's Law (time dependence for latitudes) and a time independence of longitudes. This extends the results of \citet{garcia2020optimal} which says that the longitude distribution is invariant when averaged over the entire solar cycle. This experiment was repeated for each of the other solar cycles. The confidence regions for the maximal symmetry of $\EE( t \mid X)$ are shown in table \ref{tab:sun_spot_syms}. All except for four cycles had confidence regions that contained the symmetry $G_{15} = S^1_z$, and these four cycles rejected all symmetries in the lattice base. This means for these cycles, the longitudinal distribution changes over time, as well as the latitudinal changes known in Sp\"{o}rer's Law.     \\

\begin{table}
\caption{ \label{tab:sun_spot_syms} Confidence Regions for Each Solar Cycle}
\fbox{%
\begin{tabular}{ c | l }
\textbf{Solar Cycle} & \textbf{Confidence Region} \\ \hline
11 & $\{ I \}$ \\
12 & $\{ I \}$ \\
13 & $\{ I, G_{15} \}$ \\
14 & $\{ I, G_2, G_7, G_8, G_{15} \}$ \\
15 & $\{ I, G_{15} \}$ \\
16 & $\{ I, G_{15} \}$ \\
17 & $\{ I \}$ \\
18 & $\{ I \}$ \\
19 & $\{ I, G_7, G_{15} \}$ \\
20 & $\{ I, G_1, G_7, G_{15} \}$ \\
21 & $\{ I, G_{15} \}$ \\
22 & $\{ I, G_2, G_7, G_8, G_{15} \}$ \\
23 & $\{ I, G_1, G_2, G_3, G_7, G_8, G_9, G_{11}, G_{12}, G_{14}, G_{15} \}$ 
\end{tabular} }
\end{table}

\section{Discussion}
\label{sec:con}

We have given methods for finding and using symmetry in regression function estimation. We have shown that we can generalise the idea of variable subset selection into selection of non-linear combinations of the variables by placing them within the framework of subgroup lattices. We have developed methods for testing and used these to estimate the maximal invariant symmetry of a regression function. We have demonstrated the power and applicability of these tests across several numerical experiments, showing that when the symmetry exists our methods do detect this consistently and that this can significantly improve estimates of the regression function.  \\
\\
The main limitation of this method is that $\hat{G}_\ell$ is biased towards more symmetry because it relies on a hypothesis test. This is deliberate here - it gives a more consistent estimate of $\hat{G}$ for precisely the reason above. It is also much faster computationally than creating a model of $f$ first, and doesn't rely on doing a good job of the modelling, especially in the high dimensional case. However, it may be a detriment when using $\hat{G}$ to smooth $\hat{f}$. It is worth considering methods akin to forward selection and backwards elimination in the variable selection literature. One could, for example, look to minimise residual sum of squares of the smoothed models $\hat{f}_G$ over the lattice $K$ at each level as with forward selection.\\
\\
Another limitation here is the existence of the symmetries of $f$. It may be the case that $f$ is close (in $L^2$ distance) to an invariant function but not one itself. In this case estimation of $G_{\max}(f)$ is at least as hard as distinguishing between the two possible functions. However, using the symmetry can still improve the estimate of $f$ by adding bias (smoothing from the symmetry) for much reduced variance. We have shown that our estimate is mildly robust to such a situation in example \ref{eg:est_f_sims}. It would be an interesting question of further study to quantify these effects.\\
\\
Overall, given the increasing use of symmetry based methods in statistics and machine learning, the methods presented in this paper should allow the full exploitation, while also guarding against erroneous over-assumption, of symmetries in the data.

\bibliography{rss_bib.bib}
\bibliographystyle{rss}

\newpage
\appendix

\section{Proofs and further mathematical details}

\subsection{Proofs in Section \ref{sec:background} }

\begin{proof}[Proof of Proposition \ref{prop:closed_sublattice}]
We need only show that $K(\mathcal{G})$ contains supremums and intersections (with respect to the subgroup ordering) of arbitrary pairs $G,H \in K(\mathcal{G})$, and that these are given by the join and meet operations in the statement. If $G,H \in K(\mathcal{G} )$, then consider that $G \cap H$ is still the infimum as with the subgroup lattice. For the supremum, consider that $G, H \leq \overline{ \langle G, H \rangle } $ by definition, so this is an upper bound. If $G, H \leq A \in K(\mathcal{G} )$ then $\langle G, H \rangle \leq A$ and so $ \overline{ \langle G, H \rangle  } \leq \overline{A} = A $, and so $\overline{ \langle G, H \rangle }$ is in fact the least upper bound and we are done.
\end{proof}

\begin{proof}[Proof of Proposition \ref{prop:lattice_as_vees}]
	First, note that $\langle G_1, \dots G_\ell \rangle_{K(\mathcal{G})}$ must contain all groups of the specified form by the definition of a lattice. Thus we need to show that the set of such groups forms a lattice. Let $G = G_{i_1} \vee \cdots \vee  G_{i_m}$ and $H = G_{j_1} \vee \cdots \vee G_{j_b}$ be two groups in the set considered. The supremum $G \vee H$ is also of the form specified because of the commutativity, associativity, and idempotency of the operation $\vee$ \cite{schmidt2011subgroup}. \\
	\\
	For the infimum, first define $\mathcal{I} = \{ i \in [\ell] : G_i \leq G \cap H \}$, and let $A = \bigvee_{i \in \mathcal{I}} G_i$. See that $A \leq G \cap H$ by definition. Take any element $g \in G \cap H$, and take sequences $\{ g^G_n \}_{n \in \NN} \subseteq \langle G_{i_a} : a \in [m] \rangle$ and  $\{ g^H_n \}_{n \in \NN} \subseteq \langle G_{j_a} : a \in [a] \rangle$ both convergent to $g$. If there is a subsequence of either in $\langle G_{i} : i \in \mathcal{I} \rangle$ then $g \in A$. If not, then by the pigeon hole principle there must exist a subsequence of $\{ g^G_n \}$ contained in some $G_{i_a}$ for $i_a \notin \mathcal{I}$, and similarly for $\{ g^H_a \}$ in some $G_{j_{a'}}$ for $j_{a'} \notin \mathcal{I}$. But then $g \in G_{i_a} \cap G_{j_{a'}}$ as both are closed groups. However, since $i_a, j_{a'} \notin \mathcal{I}$ we know that $G_{i_a} \cap G_{j_{a'}} = I$ (as if one were contained in the other then the index of the smaller group would be in $\mathcal{I}$), and so $g = e \in A$. Therefore $A = G \cap H$ and so the set described contains all infimums and so is a lattice. 
\end{proof}


The proof of Proposition \ref{prop:max_subgroup} result requires a few supporting lemmas.


\begin{Lem}
\label{lem:generate_inv}
	Let $\mathcal{G}$ be a group with subgroups $G,H \leq \mathcal{G}$.	Suppose that $f : \mathcal{X} \rightarrow \RR$ is $G$-invariant. If $H \leq G$ then $f$ is $H$-invariant. If $f$ is $H$ invariant then $f$ is $\langle G, H \rangle$-invariant.
\end{Lem}

\begin{proof}[Proof of Lemma \ref{lem:generate_inv}]
The first claim is clear: if $f(g \cdot X) = f(X)$ for all $g \in G$ almost surely then it is true for all $g \in H \subseteq G$. For the second, take $x \in \mathcal{X}$ and $a \in \langle G, H \rangle$. Express $a = a_1^{i_1} \cdots a_k^{i_k}$, where $a_j \in G$ or $a_j \in H$ and $i_j \in \NN$ for all $j \in [k]$. Then see that
\begin{equation}
f( a \cdot x ) = f ( a_1^{i_1} \cdot ( a_2^{i_2} \cdots a_k^{i_k} \cdot X ) ) = f( a_2^{i_2} \cdots a_k^{i_k} \cdot X )
\end{equation}
almost surely so the second claim follows by an induction on $k$.
\end{proof}

\begin{Lem}
\label{lem:lattice_inv}
Suppose that $f$ is not $G$-invariant for some subgroup $G \leq \mathcal{G}$. Then $f$ is not $H$-invariant for any $H$ with $G \leq H \leq \mathcal{G}$. If $f$ is $H$ invariant then $f$ is not $G'$ invariant for any $G'$ with $\langle G', H \rangle = \langle G, H \rangle$.
\end{Lem}

\begin{proof}[Proof of Lemma \ref{lem:lattice_inv}]
	The first result is just the contraposition of Lemma \ref{lem:generate_inv}. The second part also follows as if $f$ were $G'$ invariant, then it would be $ G \leq \langle G' , H \rangle = \langle G, H \rangle $ invariant.
\end{proof}

We can now prove Proposition \ref{prop:max_subgroup}.
\begin{proof}[Proof of Proposition \ref{prop:max_subgroup}]
	Let $G_{\max} = \langle H \leq G : f \text{ is } H \text{-invariant} \rangle$, noting that $f$ is $I$-invariant and so there is at least one such $H$. Then the previous lemma gives the first condition of maximal invariance, and the second is trivial and implies also uniqueness.
\end{proof}

\begin{proof}[Proof of Proposition \ref{prop:rules}]
	These rules are simply a restating of the Lemmas \ref{lem:generate_inv} and \ref{lem:lattice_inv}.
\end{proof}

\begin{proof}[Proof of Proposition \ref{prop:sub_lattice_restriction}]
	This follows from exactly the same reasoning as Proposition \ref{prop:max_subgroup}, except here we need only consider finitely many groups $H$ and have to take a closure operator: i.e.
	\begin{equation}
		A = \overline{ \langle H \in K :  f \text{ is } H \text{-invariant} \rangle }
	\end{equation}
	Since $I \in K$, we know $A \geq I$. Since $K$ is a finite lattice, $A$ must also be in $K$. Lastly note that Lemma \ref{lem:generate_inv} implies that $f$ is $A$-invariant and so $A \leq \Gmax$.
\end{proof}


\begin{Lem}
	\label{lem:closure_invariance}
	Let $G$ be a topological group. Suppose that $f : \mathcal{X} \rightarrow \RR$ is continuous and $\overline{G}$ acts continuously on $\mathcal{X}$. Then $f$ is $G$-invariant if and only if $f$ is $\overline{G}$-invariant.
\end{Lem}
\begin{proof}[Proof of Lemma \ref{lem:closure_invariance}]
	Since $G \subseteq \overline{G}$ the forwards direction is clear. Suppose that $f$ is either weakly or strongly $G$-invariant. Take $g \in \overline{G}$ and a net $g_\bullet$ in $G$ converging to $g$, and note that
	\begin{equation}
		f(g \cdot X) = f( ( \lim g_\bullet) \cdot X) = f ( \lim g_\bullet \cdot X ) = \lim f( g_\bullet \cdot X) = \lim f(X) = f(X)
	\end{equation}
	almost surely, giving the result.
\end{proof}

%
%

\subsection{Proofs in Section \ref{sec:lattice}}

\begin{proof}[Proof of Proposition \ref{prop:conf_region}]
	We check the conditions (R1) - (R3) in order. The first is trivial because each test is a measurable function. \\
	\\
	For the second condition, suppose that $G \in \hat{K}_\ell$ has subgroup $H$. Then if $G_i \leq H$, we must have $G_i \leq G$. Thus we must have $\mathrm{Test}_{\alpha_i}( \mathcal{D} , G_i ) = 1$ for all $G_i \leq H \leq G$, so $H \in \hat{K}_\ell$. \\
	\\
	Since $\hat{K}_\ell$ always contains the trivial group $I$ we need only consider the case where $G_{\max}(f, K) > I$. Let $\mathcal{I} = \{ i \in [\ell] : G_i \leq G_{\max}(f, K) \}$. If $G_{\max}( f, K)$ is not contained in $\hat{K}_\ell$, then at least one of the hypotheses $H^{(i)}_0$ is rejected for some $i \in \mathcal{I}$. Thus:
	\begin{align}
	 	\PP\big( G_{\max}(f, K) \not\in \hat{K}_{\ell} \big) &= \PP\big( \mathrm{Test}_{\alpha_i} ( \mathcal{D}, G_i ) = -1 \text{ for some } i \in \mathcal{I} \big) \\
	 	&\leq \sum_{ i \in \mathcal{I} } \PP \big( \mathrm{Test}_{\alpha_i} ( \mathcal{D}, G_i ) = -1 \big) \\
	 	&\leq \sum_{ i \in \mathcal{I} } \alpha a_i \leq \sum_{i = 1}^\ell \alpha_i \leq  \sum_{ i \in \NN} \alpha_i = \alpha
	\end{align}
	which gives condition (R3). 
\end{proof}

\begin{proof}[Proof of Theorem \ref{thm:cons_of_hat_G}]
	First note that $\PP_f( \mathrm{Test}_{\alpha_i} ( \mathcal{D}, G_i ) = 0 ) \rightarrow 1$ for all $G_i \not\leq G_{\max}(f, K_\ell)$ by the assumption of consistency of the test. Thus:
	\begin{align}
		\PP_f( \hat{G}_\ell \not\leq G_{\max}( f, K_\ell ) ) &\leq \PP_f ( \mathrm{Test}_{\alpha_i} ( \mathcal{D}, G_i ) = 1 \text{ for some } G_i \not\leq G_{\max}(f, K_\ell) ) \\
		&\leq \sum_{ G_i \not\leq G_{\max}(f, K_\ell) } 1 - \PP_f( \mathrm{Test}_{\alpha_i} ( \mathcal{D}, G_i ) = 0 ) \rightarrow 0
	\end{align} 
	giving the first statement. Now let $P_\alpha$ be a lower bound on the power of the consistent test used, which must be an increasing sequence (with respect to $n$) for each fixed $\alpha$ converging to $1$. We first construct the sequences of significance levels. Simply take some fixed starting points $\alpha_i(1)$ for $i \in \{ 1, \dots, \ell \}$ and set $\alpha_i(n) = \alpha_{i}(n-1)$ until $P_{\alpha_i(n)} \geq 1 - \alpha_i(n)/2$, then set $\alpha_{i}(n) = \alpha_i(n-1) / 2$. With this sequence of significance levels, we have the convergences:
	\begin{align}
		\PP_f( \mathrm{Test}_{\alpha_i(n)} ( \mathcal{D}, G_i ) = 1 ) &\rightarrow 1 & \forall G_i \leq G_{\max}(f, K_\ell) \\
		\PP_f( \mathrm{Test}_{\alpha_i(n)} ( \mathcal{D}, G_i ) = 0 ) &\rightarrow 1 & \forall G_i \not\leq G_{\max}(f, K_\ell)
	\end{align}
	From these it is clear that $\mathrm{TestResults}$ (as defined in algorithm \ref{algo:comp_K}) converges in probability to the vector of subgroups of $G_{\max}(f, K_\ell)$. \\
	\\
	This then implies that $\PP_f( G \in \hat{K}_\ell ) \rightarrow 0$ for each $G \notin \{ G \in K_\ell : G \leq G_{\max}(f, K_\ell) \}$, so we have the convergence in probability of $\hat{K}_\ell$ to $\{ G \in K_\ell : G \leq G_{\max}(f, K_\ell) \}$. Since this set has only one maximal element, $G_{\max}(f, K_\ell )$, we know that $\hat{G}_\ell$ converges in probability to this.
\end{proof}

\subsection{Proofs in Section \ref{sec:testing}}
\label{app:proof4}

\begin{proof}[Proof of Lemma \ref{lem:test_stat_control}]
Under the null hypothesis of $G$-invariance, we know that 
\begin{align}
\label{eq:like_bound}
	|Y_i - Y_j | &= | f(X_i) - f(X_j) + \epsilon_i - \epsilon_j | \\
		&= |f (g \cdot X_i) - f(X_j) + \epsilon_i - \epsilon_j | \\
		 &\leq |f(g \cdot X_i) - f(X_j)| + |\epsilon_i - \epsilon_j| \\
	&\leq V(g \cdot X_i, X_j) + |\epsilon_i - \epsilon_j|,
\end{align}
so we know that $D_{ij}^g = |Y_i - Y_j | - V( g \cdot X_i, X_j) \leq |\epsilon_i - \epsilon_j|$ for all $i,j \in \{1, \dots, n \}$ and $g \in G$. In particular, this is true for $j$ chosen such that $d( g \cdot X_i, X_j)$ is minimised. If we then count how often $D_{ij}^g$ is larger than some threshold $t \in \RR_{\geq 0
}$, we can bound the probability using the concentration of $| \epsilon_i - \epsilon_j |$. \\
\\
For fixed $m$, the $D_{I(j) J(j)}^{g_j}$ are asymptotically independent (because it is vanishingly unlikely that we sample the same $g_j \cdot X_{I(j)}$ or that two $g_j \cdot X_{I(j)}$ share a nearest neighbour). Thus under the null hypothesis, $N_t^g$ is eventually stochastically bounded by a $\mathrm{Binom}(m, p_t)$ variable.
\end{proof}

We require the following result to simplify the sampling distributions when testing for symmetry in section \ref{sec:testing}.
\begin{Lem}
	\label{lem:topo_gens_invariance}
	Suppose $f : \mathcal{X} \rightarrow \RR$ is continuous and $G$ acts continuously on $\mathcal{X}$.	If $G$ is topologically generated by a finite set $\{ g_1, \dots, g_k \}$ and $f$ is not $G$-invariant, then $f$ is not invariant to at least one group generated by a topological generator $\langle g_i \rangle$.
\end{Lem}

	\begin{proof}[Proof of Lemma \ref{lem:topo_gens_invariance}]
	Let $g \in G$ be such that $\PP( f( g \cdot X) = f(X) ) < 1$, and let $\{ h_i \}_{i = 1}^\infty$ be a sequence in $\langle g_1, \dots, g_k \rangle$ that approximates $g$. Suppose for the sake of contradiction that $f$ is invariant to every $\langle g_i \rangle$. Then $f( h_i \cdot X ) = f(X)$ almost surely (as $h_i$ is a finite product of only $g_i$ terms) for every $i \in \NN$. But then the continuity of $f$ (noting that $\mathcal{X}$ is a first countable metric space) implies
	\begin{equation}
		f(X) = \lim_{i \rightarrow \infty} f( h_i \cdot X) = f( \lim_{i \rightarrow \infty} h_i \cdot X) = f(g \cdot X)
	\end{equation}
	almost surely - a contradiction.
\end{proof}

The proof of proposition 4.1 relies on several supporting Lemmas. In the following we use the notation $X \overset{\PP}{\rightarrow} x$ to indicate that the random variable $X$ converges in probability to $x$.

\begin{Lem}
In the context of proposition \ref{prop:cons}, if the support of $\mu_X$ is dense in $\mathcal{X}$ then
\begin{equation}
d_\mathcal{X} ( g \cdot X_{I(1)}, X_{J(1)} ) \overset{\PP}{\rightarrow} 0	
\end{equation}
\end{Lem}

\begin{proof}
	For any $\eta \in \RR_{\geq 0 }$,  the probability that $d_\mathcal{X} ( g \cdot X_{I(1)}, X_{J(1)} ) \geq \eta$ given $I, g$ is equivalent to a binomial variable $K$ with $n$ trials and probability of success $p = \mu_X( \mathcal{X} \setminus B (g \cdot X_{I(1)}, \eta) )$ being $0$. Since the support of $\mu_X$ is dense in $\mathcal{X}$, $\mu_X( B (g \cdot X_{I(1)}, \eta) ) > 0$, and so
	\begin{equation}
		\PP( d_\mathcal{X} ( g \cdot X_{I(1)}, X_{J(1)} ) \geq \eta ) \leq \PP( K = 0 ) = (1 - p)^n
	\end{equation}
	 As $n \rightarrow \infty$, this clearly goes to $0$ for all $\eta$, and for all $I, g$.
\end{proof}

In the following let $F_X(x)$ be the distribution function of the real valued random variable $X$.

\begin{Lem}
	\label{lem:sym_props}
	If $X$ and $Y$ are independent real random variables, and $Y$ is symmetrically distributed around $0$ (so $F_Y(t) = 1 - F_Y(-t)$ for all $t \in \RR$) and admits a density $f_Y$ with respect to Lebesgue measure on $\RR$ that is decreasing in $|y|$. Then $\PP( |X + Y| \geq t) > \PP( |Y| \geq t )$.
\end{Lem}

\begin{proof}
	First note $\PP( X + Y \geq t) = \EE_X (\PP( Y \geq t - X \mid X) ) = 1 - \EE_X( F_Y(t- X ) )$. Thus
	\begin{align}
		\PP( |X + Y |\geq t ) &= \PP( X + Y \geq t) + \PP( X + Y \leq -t ) \\
			&= 2 - \EE_X( F_Y(t - X) + F_Y(t + X) )
	\end{align}
	Let $\phi_t(x) = F_Y(t - x) + F_Y(t + x)$. Since $Y$ is absolutely continuous (w.r.t. Lebesgue measure) with density function $f_Y$, we can see that $\phi_t$ has a critical point at $x = 0$. Moreover,	if $x > 0$ then
	\begin{equation}
		\phi'_t(x) = f_Y(t + x) - f_Y(t - x) = f_Y(t + x ) - f_Y(-t + x ) <  f_Y(t + x ) - f_Y(-t - x) = 0
	\end{equation}
	as $f_Y$ is decreasing in $|y|$. Similarly $\phi'_t(x) > 0$ if $x < 0$. Thus $x = 0$ is a maximum and so we have
	\begin{equation}
		\PP( |X + Y |\geq t ) = 2 - \EE_X( F_Y(t - X) + F_Y(t + X) ) \geq 2 - \EE_X( F_Y(t) + F_Y(t) ) = \PP( |Y| \geq t ).
	\end{equation}
	as required.
\end{proof}

\begin{Lem}
	\label{lem:binom_likelihood}
	If 	$X \sim Binom(n, p)$ and $Y \sim Binom(n, q)$ with $p < q$ , then $F_X( Y) \overset{p}{\rightarrow} 1$
\end{Lem}

\begin{proof}
	First note that $F_X(Y) = F_{ \tilde{X} } ( \tfrac{Y - n p}{ \sqrt{n p (1 - p) } } )$ where $\tilde{X} = \tfrac{X - n p}{ \sqrt{n p (1 - p) } }$, which has $\tilde{X} \overset{D}{\rightarrow} N(0, 1)$ by the de Moivere - Laplace theorem (DLT). Then we also set:
	\begin{equation}
		\tilde{Y} = \frac{Y - n p}{ \sqrt{n p (1 - p) } } = \left( \frac{Y - n q}{ \sqrt{n q (1 - q) } } + \sqrt{n} \frac{q -  p}{ \sqrt{ q (1 - q) } } \right) \frac{\sqrt{q (1 - q)} }{ \sqrt{ p (1 - p) } }.
	\end{equation}
	Again by DLT, the first term converges to $N(0, \tfrac{q (1 - q) }{p (1 - p} )$, but the second term diverges to $+\infty$ (as $q > p$). Thus $\tilde{Y} \overset{p}{\rightarrow} \infty$ and so clearly $F_{\tilde{X} }( \tilde{Y} ) \overset{p}{\rightarrow} 1$, as required.
\end{proof}

We can now return to the proof of Proposition \ref{prop:cons}.

\begin{proof}[Proof of Proposition \ref{prop:cons}]
	
Since the support of $\mu_X$ is dense in $\mathcal{X}$, Lemma 4 gives $d_\mathcal{X} ( g \cdot X_{I(j)}, X_{J(j)} ) \overset{p}{\rightarrow} 0$. Now consider
\begin{align}
D_{I(j)J(j)}^{g_j} &= | Y_{I(j)} - Y_{J(j)}| - V( g \cdot X_{I(j)}, X_{J(j)} ) \\
	&= | f(X_{I(j)} ) + \epsilon_{I(j)} - f(X_{J(j)}) - \epsilon_{J(j)} | - V( g \cdot X_{I(j)}, X_{J(j)} )
\end{align}
We have that $V( g \cdot X_{I(j)}, X_{J(j)} ) \rightarrow 0$ as the distance goes to $0$, and also 	$f(X_{J(j)}) - f( g \cdot X_{I(j)}) \overset{p}{\rightarrow} 0$. Thus with $X$ as a variable with the same law as each $X_i$ and $g$ with the same law as each $g_j$,
\begin{equation}
	D_{I(j)J(j)}^{g_j} \overset{D}{\rightarrow} | (f(X) + \epsilon_\ell) - f(g \cdot X) - \epsilon_k  | = | \phi(g, X) + \eta |	
\end{equation}
where $\phi(g, X) = f(X) - f ( g \cdot X )$ and $\eta = \epsilon_l - \epsilon_k$. Definitionally, this means
\begin{equation}
\PP( D_{I(j) J(j) }^{g_j} \geq t ) \rightarrow \PP( | \phi(g, X) + \eta | \geq t ) \end{equation}
Under $H_0$, $\phi(g,X) = 0$ almost surely, so this probability is given by $P_t^0 = \PP( |\eta| \geq t)$. Under $H_1$, $\phi$ must take non zero values with some positive probability (using the condition on the distribution $\mu_G$). Thus we can use Lemma \ref{lem:sym_props} to say that
\begin{equation}
	P_t^1 = \PP(  | \phi(g, X) + \eta | \geq t \mid H_1 ) > P_t^0
\end{equation}
Let $N$ be large enough that $\PP( D_{I(j) J(j) }^{g_j} \geq t ) > (P^1_t - P^0_t) / 2 > P_t^0$ for all $n > N$.  \\
\\
Now consider that under the alternative hypothesis, $N_t^g \overset{D}{\rightarrow} \mathrm{Binom}( m, \PP( D_{I(j) J(j) }^{g_j} \geq t )  )$, which is stochastically bounded from below by $A \sim \mathrm{Binom} (m, (P_t^1 - P_t^0)/2) $ for $n > N$. This gives, using lemma \ref{lem:binom_likelihood}, with $A' \sim \mathrm{Binom}( m, P_t^0 )$, that
\begin{equation}
	p_{val} = 1 - F_{ A' }( N_t^g ) \leq 1 - F_{ A } ( N_t^g ) \overset{p}{\rightarrow} 0
\end{equation}
as required.
\end{proof}

\section{Estimating Lipschitz Constant of $f$}
\label{app:estimate_lipschitz_const}

As noted in section \ref{ssec:asym_var_test}, if the bound $V$ is not assumed then we may need to estimate it from the data. In the same vein as this test a natural estimate for $V$ is $\hat{L} d(x,y)$, where $\hat{L}$ is the smallest positive $L$ such that $N^e_t \leq Q_{m, p_t}(\alpha)$ where $N_t^e$ is the same tests statistic before but where $\mu_g$ only has support on the identity $e$ and $Q_{m, p_t}(\alpha)$ is the $\alpha$-quantile of a binomial variable with $m$ trials and probability of success $p_t$. \\
\\
Of course, such an estimate is very much not independent of the asymmetric variation test so this could be estimated using a data-splitting approach if needed.

\section{Further Details on the Experiments and Applications}

\subsection{High dimensions for example \ref{eg:low_dim_G_sims}}
\label{app:higher_dimension_eg_61}

As in example \ref{eg:low_dim_G_sims}, we plot the confidence region excess $\hat{\mathcal{E}} ( \hat{K}_\ell )$ for dimensions 10 and 20 in figure \ref{fig:high_dimensional_plots}. We see that the asymmetric variation test is unable to reject any symmetries in these higher dimensional settings, but that the permutation variant can. This is remarkable because estimating anything about a nonparametric $f$ in dimension 10 is very hard with so few data points. 

\begin{figure}[h]
	\centering
	\begin{tabular}{ m{7cm} m{7cm} }
		\centering \includegraphics[scale=0.48]{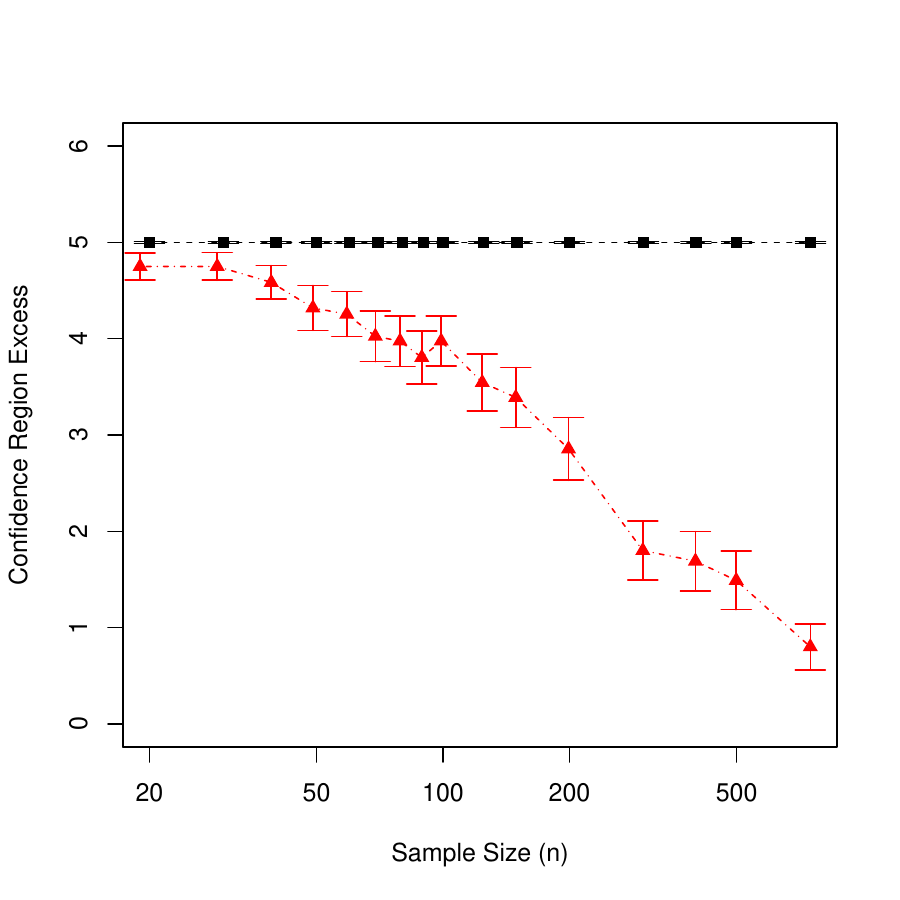} &
		\centering \includegraphics[scale=0.48]{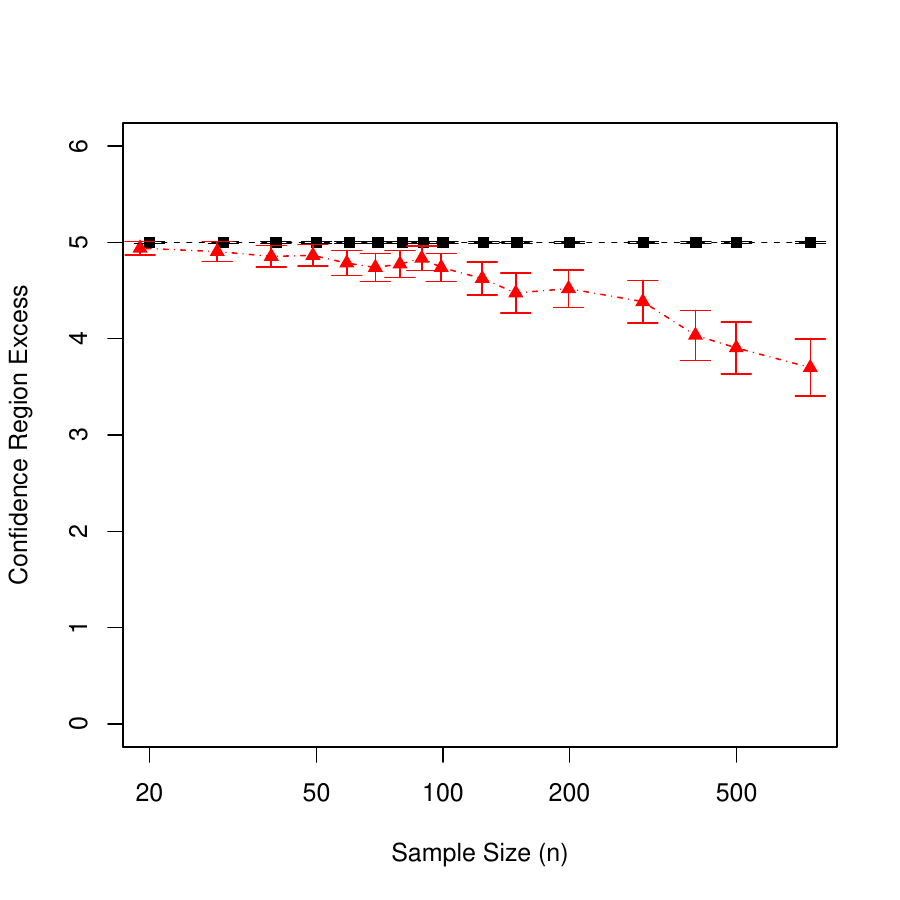} \tabularnewline
		\centering (a) $\mathcal{E}( \hat{K}_\ell )$ when $d = 10$ & 
		\centering (b) $\mathcal{E}( \hat{K}_\ell )$ when $d = 20$
	\end{tabular}
	\caption{A replica of figure \ref{fig:toy_eg_K_hat} for dimensions $d = 10$ and $d= 20$. }
	\label{fig:high_dimensional_plots}
\end{figure}

\subsection{Effect of using a fixed $\alpha_i$ in example \ref{eg:low_dim_G_sims} }
\label{app:correct_containment}

As mentioned in section \ref{ssec:choose_alpha_i} and in the variable selection literature \citep{heinze2018variable}, one might choose to conduct all tests at the same significance level, i.e., $\alpha_i = a$. We can use the experiments in  example $\ref{eg:est_f_sims}$ to examine how this choice affects the confidence region. In figure \ref{fig:correct_containments} we plot the proportion of the confidence regions that contain the true maximal symmetry $G_{\max}(f, K_\ell)$. Each test was conducted at the significance level $\alpha_i = 0.05$, which guarantees that $\hat{K}_\ell$ is a $0.7$-confidence region for $G_{\max}$. However, we also see that $G_{\max} \in \hat{K}_\ell$ much more frequently, and for $n \geq 30$, this is in more than $1 - \alpha_i$ of the simulations. 

\begin{figure}[h]
	\centering
	\begin{tabular}{m{7cm} m{7cm}}
		\centering \includegraphics[scale = 0.42]{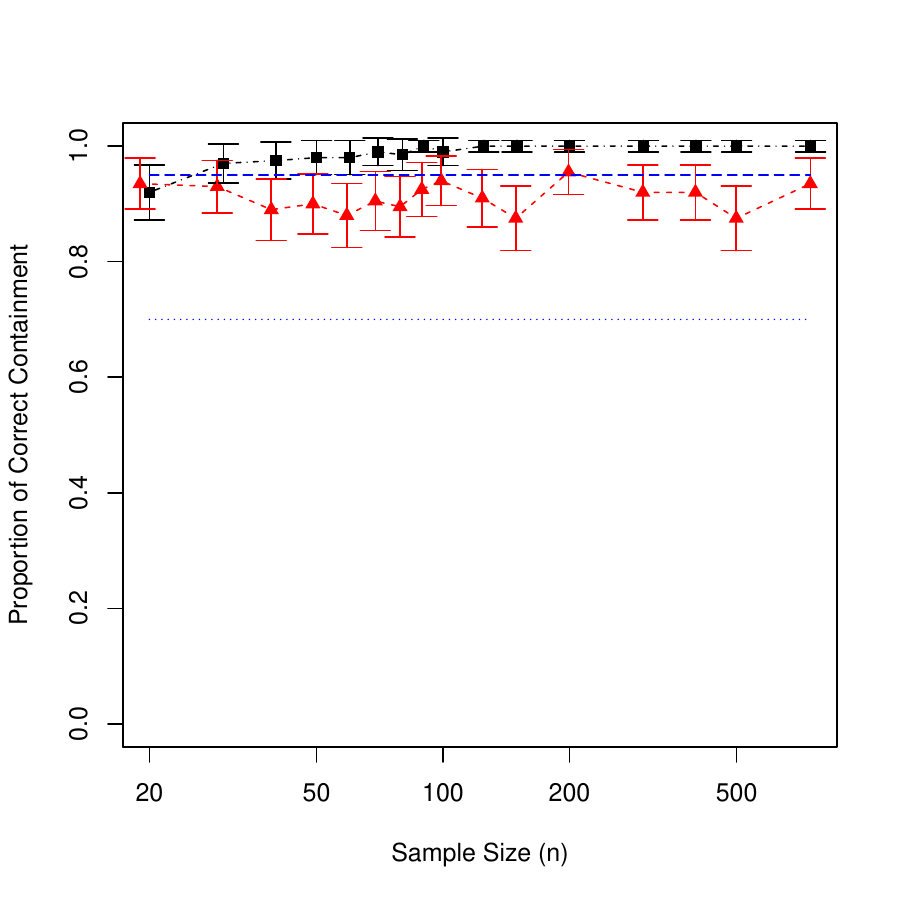} & 
		\centering \includegraphics[scale = 0.42]{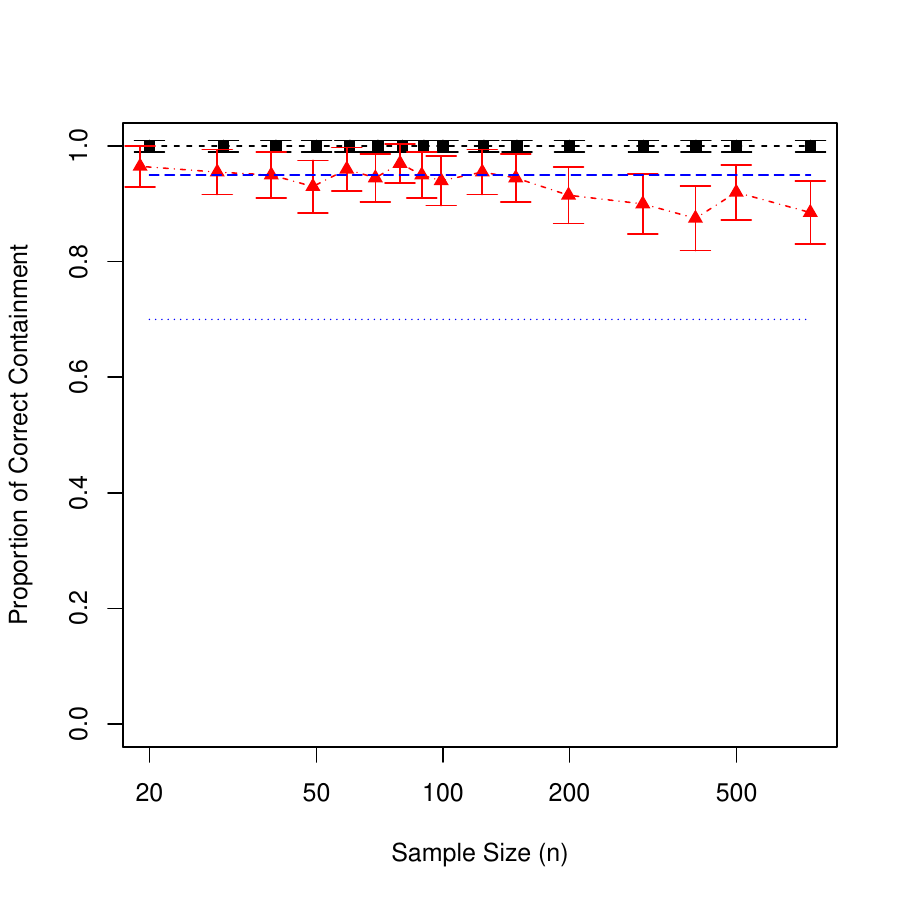} \tabularnewline
		\centering (a) Dimension 2 &
		\centering (b) Dimension 4
	\end{tabular}
	\caption{Plot of the frequency with which $G_{\max}(f, K_\ell) \in \hat{K}_\ell$ in the 100 simulations of example \ref{eg:low_dim_G_sims}. The black squares represent the results when using the asymmetric variation test and the red triangles for the permutation variant. The lower and upper blue lines represent the 0.7 and 0.95 confidence region error thresholds respectively.}
	\label{fig:correct_containments}
\end{figure}

\subsection{Effects of changing the lattice base}

\label{app:lattice_base_effects}

Changing the size of the lattice base will have several effects. As we include more tests to conduct we certainly increase the variance in which group we estimate, but potentially improve the estimation of the regression function (by reducing bias with a more suitable group, if there was not a true symmetry in the original lattice). We examine these effects empirically in the context of example \ref{eg:est_f_sims}. \\
\\
We take four sizes of the lattice base, the first only containing the groups $S^1_{u_i}$ for $i \leq 6$ (again the vectors are given in appendix \ref{app:icosahedral_vectors}). The second also contains the second set $S^1_{u_i}$ for $i \leq 12$. The third also adds $S^1_{z}$, and the fourth adds in $S^1_x$ and $S^1_y$. \\
\\
Scenarios 2 and 4 contain symmetries isomorphic to $S^1$, but around different axes (in the first case, the $z$-axis and the second an axis offset from the $z$-axis). Thus the third and fourth lattice bases contain either true symmetries or approximate symmetries, but the first two do not. \\
\\
We run the simulations as in example \ref{eg:est_f_sims}, and plot the results in figure \ref{fig:effects_of_varying_K}. We see, as expected, that smaller lattices result less variance in the MSPE, but that this effect is dwarfed by the effect of whether $K$ contains an approximate symmetry or not. This suggests that in practice we should typically include as many groups in $K$ as we are willing to compute tests for. 

\begin{figure}[h]
    \centering
    \begin{tabular}{m{7cm} m{7cm}}
		\centering \includegraphics[scale = 0.42]{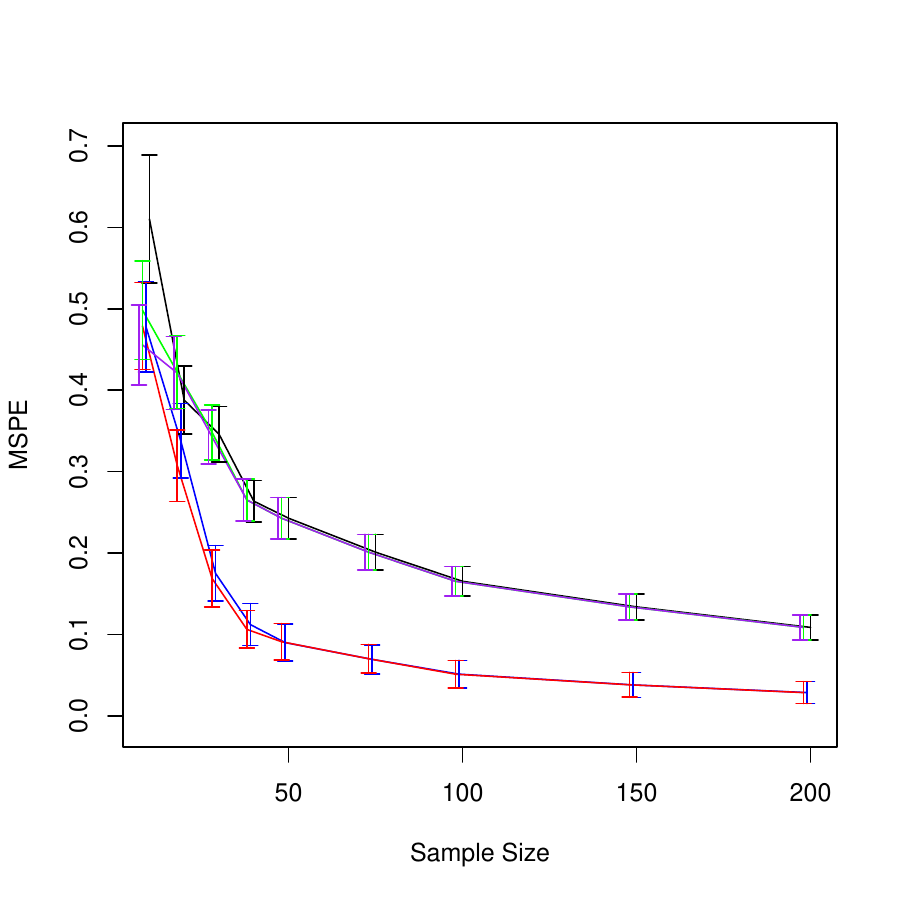} & 
		\centering \includegraphics[scale = 0.42]{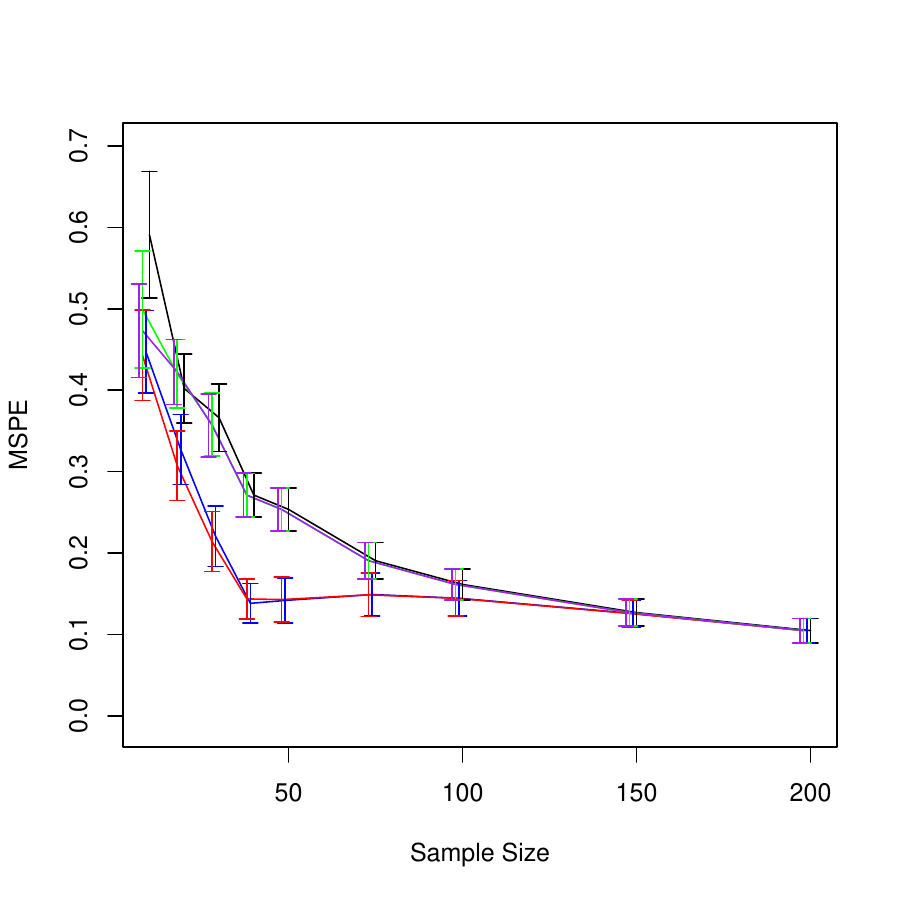} \tabularnewline
		\centering (a) Scenario 2 &
		\centering (b) Scenario 4
	\end{tabular}
    \caption{Means square predictive errors for various sizes of the lattice base, in scenarios 2 and 4. The black lines are the means errors for the baseline LCE across 100 samples. The purple corresponds to using the smallest lattice $K_{6}$, then green for $K_{12}$, red for $K_{13}$, and lastly blue for $K_{15}$. Also shown are 95\% Wald confidence intervals for the true means. }
    \label{fig:effects_of_varying_K}
\end{figure}

\subsection{Power and size of the hypothesis tests}
\label{app:test_results}

We now provide the results of each test conducted in example \ref{eg:low_dim_G_sims} in order to gauge the empirical size and power of the asymmetric variation test and the permutation variant. These are summarised in tables \ref{tab:prop_rejections_avt} and \ref{tab:prop_rejections_pv}. \\
\\
Recall that for each dimension $d$, the function $f_d$ is invariant to the groups $G_1 = \langle R_h \rangle$, $G_2 = \langle R_v \rangle$, and $G_3 = \langle R_{\pi} \rangle$ and is not invariant to $G_4 = \langle R_/ \rangle $, $G_5 = \langle R_\backslash \rangle $, or $G_6 = \langle R_{\pi / 2} \rangle$. 

\begin{table}
    \caption{ \label{tab:prop_rejections_avt} Proportions of Rejections for the asymmetric variation test in example 6.1.}
    \fbox{
    \begin{tabular}{c | c c |c c c c c c c c}
         \textbf{Dimension} & \textbf{Hypothesis} & \textbf{Group} & $n = $ \textbf{20} & \textbf{50} & \textbf{100} & \textbf{200} & \textbf{500}  \\ \hline
         \multirow{6}{*}{2} & \multirow{3}{*}{$H_0$} & $\langle R_h \rangle$ & 0.030 & 0.005 & 0  &  0  &  0\\
                              &                          & $\langle R_v \rangle$ & 0.030 & 0.005 & 0.005 & 0 & 0 \\
                              &                          & $\langle R_\pi \rangle$ & 0.030 & 0.010 & 0.005 & 0 & 0\\ \cline{2-8}
                              & \multirow{3}{*}{$H_1$} & $\langle R_/ \rangle$ & 0.400 & 0.810 & 0.975 &  1  & 1 \\
                              &                          & $\langle R_\backslash \rangle$ & 0.400 & 0.825 & 0.990 &   1 &   1 \\
                              &                          & $\langle R_{\pi / 2} \rangle$ &0.355 &0.785 &0.995   & 1   & 1 \\ \hline
        \multirow{6}{*}{4} & \multirow{3}{*}{$H_0$} & $\langle R_h \rangle$ & 0 & 0 & 0 & 0 & 0 \\
                              &                          & $\langle R_v \rangle$ & 0 & 0 & 0 & 0 & 0  \\
                              &                          & $\langle R_\pi \rangle$ & 0 & 0 & 0 & 0 & 0 \\ \cline{2-8}
                              & \multirow{3}{*}{$H_1$} & $\langle R_/ \rangle$ &  0.005 & 0.055 & 0.095 & 0.330 & 0.930\\
                              &                          & $\langle R_\backslash \rangle$ & 0.015 &0.010 &0.045 &0.270 &0.955\\
                              &                          & $\langle R_{\pi / 2} \rangle$ &0.015& 0.030 &0.060& 0.245& 0.870 \\ \hline 
        \multirow{6}{*}{10} & \multirow{3}{*}{$H_0$} & $\langle R_h \rangle$ & 0 & 0 & 0 & 0 & 0\\
                              &                          & $\langle R_v \rangle$ & 0 & 0 & 0 & 0 & 0 \\
                              &                          & $\langle R_\pi \rangle$ & 0 & 0 & 0 & 0 & 0\\  \cline{2-8}
                              & \multirow{3}{*}{$H_1$} & $\langle R_/ \rangle$ & 0 & 0 & 0 & 0 & 0\\
                              &                          & $\langle R_\backslash \rangle$ & 0 & 0 & 0 & 0 & 0\\
                              &                          & $\langle R_{\pi / 2} \rangle$ & 0 & 0 & 0 & 0 & 0
    \end{tabular} }
\end{table}

\begin{table}
    \caption{\label{tab:prop_rejections_pv} Proportions of Rejections for the permutation variant in example 6.1.}
    \fbox{ 
    \begin{tabular}{c | c c |c c c c c c c c}
         \textbf{Dimension} & \textbf{Hypothesis} & \textbf{Group} & $n = $ \textbf{20} & \textbf{50} & \textbf{100} & \textbf{200} & \textbf{500}  \\ \hline
         \multirow{6}{*}{2} & \multirow{3}{*}{$H_0$} & $\langle R_h \rangle$ & 0.020 & 0.020 &0.025& 0.025 &0.070\\
                              &                          & $\langle R_v \rangle$ & 0.015 &0.045 &0.030 &0.020 &0.030\\
                              &                          & $\langle R_\pi \rangle$ & 0.035 & 0.035& 0.035 &0.015 &0.060\\ \cline{2-8}
                              & \multirow{3}{*}{$H_1$} & $\langle R_/ \rangle$ & 0.585 &0.825 &0.835 &0.815 &0.830\\
                              &                          & $\langle R_\backslash \rangle$ & 0.580 &0.810 &0.805 &0.785 &0.845 \\
                              &                          & $\langle R_{\pi / 2} \rangle$ & 0.540 &0.810 &0.790& 0.770& 0.810 \\ \hline
        \multirow{6}{*}{4} & \multirow{3}{*}{$H_0$} & $\langle R_h \rangle$ &0.010 & 0.015& 0.015 &0.010& 0.030 \\
                              &                          & $\langle R_v \rangle$ & 0.010 &0.040& 0.035 &0.020 &0.045\\
                              &                          & $\langle R_\pi \rangle$ & 0.030 &0.040& 0.035 &0.055 &0.045\\ \cline{2-8}
                              & \multirow{3}{*}{$H_1$} & $\langle R_/ \rangle$ & 0.190& 0.610 &0.875 &0.955 &0.985\\
                              &                          & $\langle R_\backslash \rangle$ & 0.155 &0.605 &0.870 &0.950& 0.980\\
                              &                          & $\langle R_{\pi / 2} \rangle$ & 0.155 & 0.535 &0.875 &0.950 &0.975 \\ \hline
        \multirow{6}{*}{10} & \multirow{3}{*}{$H_0$} & $\langle R_h \rangle$ & 0.000 & 0.000 & 0.000&  0.005& 0.005\\
                              &                          & $\langle R_v \rangle$ & 0.000 & 0.005 &0.020 &0.020 &0.005\\
                              &                          & $\langle R_\pi \rangle$ & 0.000 &0.000& 0.000 &0.015 &0.010\\ \cline{2-8}
                              & \multirow{3}{*}{$H_1$} & $\langle R_/ \rangle$ & 0.040 & 0.080 &0.145 &0.275 &0.620\\
                              &                          & $\langle R_\backslash \rangle$ & 0.015 & 0.105 &0.140 &0.365 &0.620\\
                              &                          & $\langle R_{\pi / 2} \rangle$ & 0.015 &0.070 &0.130 &0.315 &0.600
    \end{tabular} }
\end{table}

\subsection{Icosahedral vertices}
\label{app:icosahedral_vectors}

In example \ref{eg:est_f_sims} our lattice base consisted of circle groups $S^1_{u_i}$ around axes corresponding to the vertices of a regular icosahedron (shown in figure \ref{fig:protiens}c). To be precise, these axes are given by the normalisation of:
\begin{align}
    u_1 &= (0, 1, \varphi) &  u_2 &= (0, 1, -\varphi) & u_3 &= (1, \varphi, 0) \\
    u_4 &= (1, -\varphi, 0) &  u_5 &= (\varphi, 0, 1) & u_6 &= (\varphi, 0, -1)
\end{align}
where $\varphi = \tfrac{1}{2}( 1 + \sqrt{5} )$ is the golden ratio. In the real world data applications (section \ref{sec:app}, we also tested for the circle groups around the additional six axes: 
\begin{align}
    u_7 &= (1, 0, \varphi) & u_8 &= (1, 0, -\varphi) & u_9 &= (\varphi, 1, 0) \\
    u_{10} &=  -\varphi, 1, 0) & u_{11} &= ( 0, \varphi, 1) & u_{12} &= ( 0, \varphi, -1)
\end{align}

\subsection{Downloading Magnetic Field Data from VirES}
\label{app:vires}

The data in the file ``SWARM\_DATA.CSV'' used in section \ref{ssec:magnets} was downloaded from the European Space Agency's VirES client (\url{https://vires.services/}). The variable names as in our csv file correspond to the variable names that can be selected using the custom download parameters, i.e., $F_{NT}$ is the magnetic field intensity and $F_{error}$ is their reported standard error on this measurement. We have downloaded only the data from satellite ``A'' on the 25th of February 2023.

\subsection{Q-Q plots in section \ref{ssec:magnets}}
\label{app:qq_plots_magnets}
	
	Here we check the appropriateness of the sampling distributions for the real data example in section 
	\ref{ssec:magnets}. As suggested in section \ref{sssec:choice_of_mu_g}, we examine the $Q$-$Q$ plots of the marginal distributions of $d( g_i \cdot X_i , g_j \cdot X_j)$ against the original $d( X_i, X_j)$ where the angle of rotation is sampled from $N(0, 0.05^2)$ and uniformly on $[0, 2 \pi]$. These plots are supplemented by the $p$-values of Kolmogorov-Smirnoff tests for the equality of these distributions. We see that the low angle distributions in figure \ref{ fig:qq_plots_all_low_angle_magnets } appear suitable with no measurable difference in the pairwise difference distributions, whilst the uniform angles in figure \ref{ fig:qq_plots_all_uniform_angle_magnets } are rejected for all groups except for $S^1_z$. In these figures the ``Group Number'' corresponds to the index of $G_i$ in the lattice base. 
	
\begin{figure}[h]
	\centering
	\includegraphics[scale=0.5]{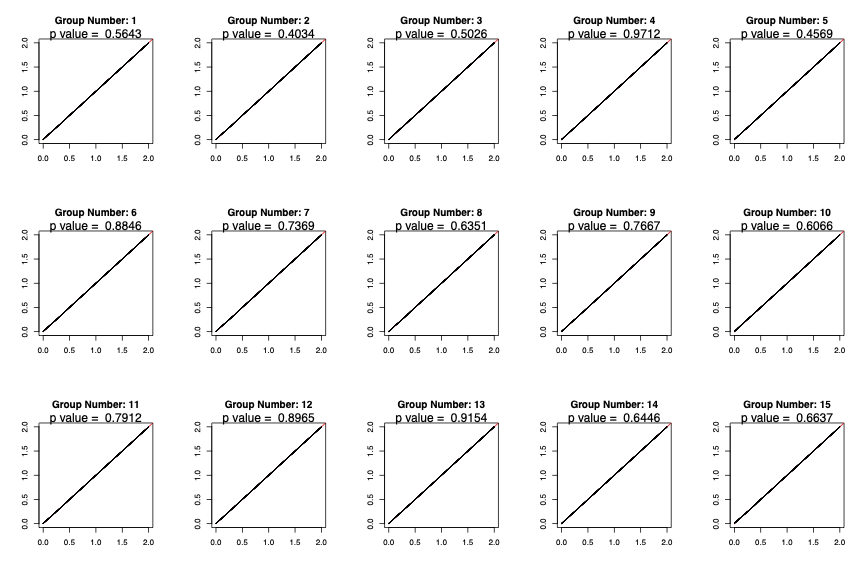}
	\caption{ Quantile-Quantile plots of $d_{\mathcal{X}}( X_i, X_j)$ vs $d_{\mathcal{X}}( g_i \cdot X_i, g_j \cdot X_j )$ for each of the groups in section \ref{ssec:sunspots}. Here the sampling of the group elements is from the low angle regime, where $g_i$ has an rotational angle sampled from $N(0, 0.05^2)$ in radians. Also presented are $p$-values for Kolmogorov-Smirnov tests for equality of the distributions.}
	\label{ fig:qq_plots_all_low_angle_magnets }
\end{figure}

\begin{figure}[h]
	\centering
	\includegraphics[scale=0.5]{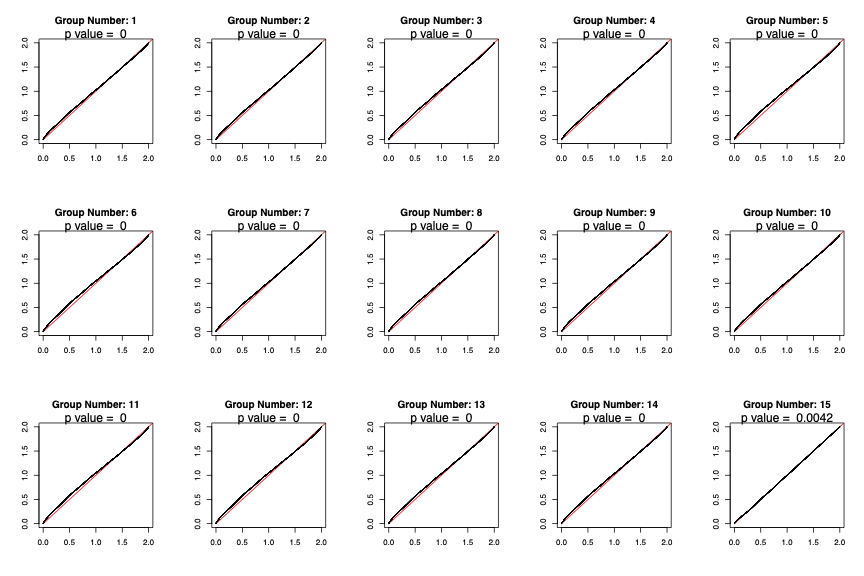}
	\caption{ Quantile-Quantile plots of $d_{\mathcal{X}}( X_i, X_j)$ vs $d_{\mathcal{X}}( g_i \cdot X_i, g_j \cdot X_j )$ for each of the groups in section \ref{ssec:sunspots}. Here the sampling of the group elements is from the uniform angle regime, where $g_i$ has an rotational angle sampled from $U( [0, 2\pi])$ in radians. Also presented are $p$-values for Kolmogorov-Smirnov tests for equality of the distributions.}
	\label{ fig:qq_plots_all_uniform_angle_magnets }
\end{figure}

\subsection{Q-Q plots in section \ref{ssec:sunspots}}
\label{app:qq_plots_sunspots}
	
	Here we check the appropriateness of the sampling distributions for the real data example in section 
	\ref{ssec:sunspots}. As suggested in section \ref{sssec:choice_of_mu_g}, we examine the $Q$-$Q$ plots of the marginal distributions of $d( g_i \cdot X_i , g_j \cdot X_j)$ against the original $d( X_i, X_j)$ where the angle of rotation is sampled from $N(0, 0.05^2)$ and uniformly on $[0, 2 \pi]$. These plots are supplemented by the $p$-values of Kolmogorov-Smirnoff tests for the equality of these distributions. We see that the low angle distributions in figure \ref{ fig:qq_plots_all_low_angle } appear suitable with no measurable difference in the pairwise difference distributions, whilst the uniform angles in figure \ref{ fig:qq_plots_all_uniform_angle } are rejected for all groups except for $S^1_z$. In these figures the ``Group Number'' corresponds to the index of $G_i$ in the lattice base. 
	
\begin{figure}[h]
	\centering
	\includegraphics[scale=0.5]{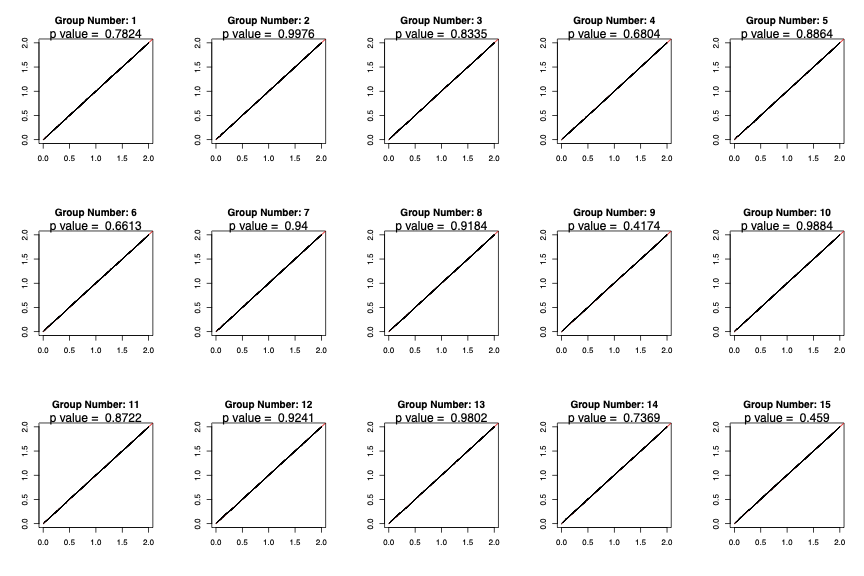}
	\caption{ Quantile-Quantile plots of $d_{\mathcal{X}}( X_i, X_j)$ vs $d_{\mathcal{X}}( g_i \cdot X_i, g_j \cdot X_j )$ for each of the groups in section \ref{ssec:sunspots}. Here the sampling of the group elements is from the low angle regime, where $g_i$ has an rotational angle sampled from $N(0, 0.05^2)$ in radians. Also presented are $p$-values for Kolmogorov-Smirnov tests for equality of the distributions.}
	\label{ fig:qq_plots_all_low_angle }
\end{figure}

\begin{figure}[h]
	\centering
	\includegraphics[scale=0.5]{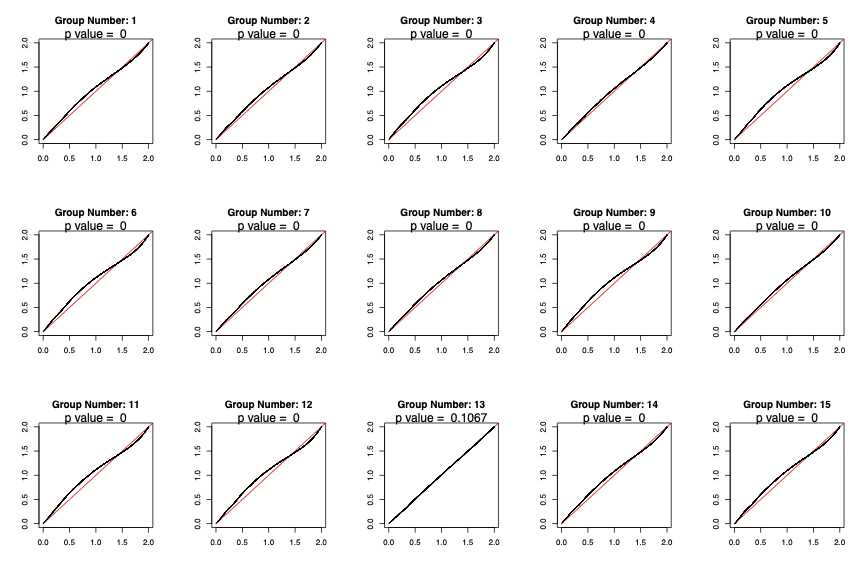}
	\caption{ Quantile-Quantile plots of $d_{\mathcal{X}}( X_i, X_j)$ vs $d_{\mathcal{X}}( g_i \cdot X_i, g_j \cdot X_j )$ for each of the groups in section \ref{ssec:sunspots}. Here the sampling of the group elements is from the uniform angle regime, where $g_i$ has an rotational angle sampled from $U( [0, 2 \pi ] )$ in radians. Also presented are $p$-values for Kolmogorov-Smirnov tests for equality of the distributions.}
	\label{ fig:qq_plots_all_uniform_angle }
\end{figure}

\end{document}